\documentclass[format=acmsmall, review=false, screen=true, manuscript=true]{acmart}
\pdfoutput=1

\usepackage{booktabs} 
\usepackage{multirow}
\usepackage{todonotes}
\usepackage{cleveref}
\usepackage{tikz}
\usepackage{subcaption}
\usepackage[ruled]{algorithm2e} 
\usepackage{enumitem}
\usepackage{color}
\usepackage{fancyvrb}
\usepackage[most]{tcolorbox}
\usepackage[mathscr]{euscript}
\usepackage{amsmath}
\usepackage{amssymb}

\SetAlFnt{\small}
\SetAlCapFnt{\small}
\SetAlCapNameFnt{\small}
\SetAlCapHSkip{0pt}
\IncMargin{-\parindent}

\acmJournal{TOMS}

\setcopyright{acmlicensed}

\acmDOI{0000001.0000001}

\definecolor{codebgcolor}{gray}{0.97}
\definecolor{indexemphcolor}{HTML}{0095e4}

\newtheorem{prop}{Proposition}
\Crefname{prop}{Proposition}{Propositions}

\newtheorem{lemma}{Lemma}
\Crefname{lemma}{Lemma}{Lemmas}

\newtheorem{ex}{Example}
\Crefname{ex}{Example}{Examples}

\newcommand{\yateto}{YATeTo}

\newcommand{\tp}[1]{{\pi\!#1}}
\newcommand{\dd}[1]{\,\mathrm{d}#1}

\DeclareMathOperator{\size}{size}
\DeclareMathOperator{\cost}{W}
\DeclareMathOperator{\nodecost}{w}
\DeclareMathOperator{\vertices}{\mathcal{V}}
\DeclareMathOperator{\children}{\mathcal{C}}
\DeclareMathOperator{\descendants}{\mathcal{D}}
\DeclareMathOperator{\granddescendants}{\mathcal{G}}

\DeclareMathOperator{\R}{\mathbb{R}}

\DeclareFontFamily{U}{mathx}{\hyphenchar\font45}
\DeclareFontShape{U}{mathx}{m}{n}{<-> mathx10}{}
\DeclareSymbolFont{mathx}{U}{mathx}{m}{n}
\DeclareMathAccent{\widebar}{0}{mathx}{"73}

\newcommand{\ms}{\mathbf{\textcolor{indexemphcolor}{s}}}

\input{code/style}

\newtcolorbox{codebox}{
  colback=codebgcolor,
  boxrule=-1pt,
  left=0.2mm,
  right=0.2mm,
  top=0.2mm,
  bottom=0.2mm
}

\RecustomVerbatimEnvironment{Verbatim}{Verbatim}{fontsize=\small}

\begin{document}
\title{Yet Another Tensor Toolbox for discontinuous Galerkin methods and other applications}

\author{Carsten Uphoff}
\affiliation[obeypunctuation=true]{%
  \institution{Technical University of Munich}
  \streetaddress{Boltzmannstr. 3}
  \postcode{85748}
  \city{Garching},
  \country{Germany}
}
\email{uphoff@in.tum.de}
\author{Michael Bader}
\affiliation[obeypunctuation=true]{%
  \institution{Technical University of Munich}
  \streetaddress{Boltzmannstr. 3}
  \postcode{85748}
  \city{Garching},
  \country{Germany}
}
\email{bader@in.tum.de}

\begin{abstract}
The numerical solution of partial differential equations is at the heart of many
grand challenges in supercomputing.
Solvers based on high-order discontinuous Galerkin (DG) discretisation
have been shown to scale on large supercomputers with excellent performance
and efficiency, if the implementation exploits all levels of parallelism
and is tailored to the specific architecture.
However, every year new supercomputers emerge and the list of hardware-specific
considerations grows, simultaneously with the list of desired features in a DG code.
Thus we believe that a sustainable DG code needs an abstraction layer to implement
the numerical scheme in a suitable language.
We explore the possibility to abstract the numerical scheme as small tensor operations,
describe them in a domain-specific language (DSL) resembling the Einstein notation, and to map them to existing code generators which generate small matrix matrix multiplication routines.
The compiler for our DSL implements classic optimisations that are used
for large tensor contractions, and we present novel optimisation techniques such as
equivalent sparsity patterns and optimal index permutations for temporary tensors.
Our application examples, which include the earthquake simulation software SeisSol,
show that the generated kernels achieve over 50\,\% peak performance while the
DSL considerably simplifies the implementation.
\end{abstract}

%
%

\ccsdesc[500]{Computing methodologies~Massively parallel and high-performance simulations}
\ccsdesc[500]{Software and its engineering~Source code generation}
\ccsdesc[500]{Software and its engineering~Domain specific languages}
\ccsdesc[500]{Applied computing~Earth and atmospheric sciences}

%
%

\maketitle

\section{Introduction}

Solving partial differential equations (PDEs) is one of the pillars of computational science and engineering,
and solving PDEs accurately on a computer is one of the grand challenges in high performance computing.
Simulations may have billions of degrees of freedom, hence highly scalable codes that make efficient
use of the invested energy are required.
But, highly efficient software often requires expert knowledge, the resulting code might not
reflect the underlying numerical scheme, and code tends to become complex:
While on the one hand a software might want to support several PDEs, several finite element types,
or multi-physics simulations, on the other hand developers have to deal with the subtleties of
modern hardware architectures, for example vector instruction latency or level 1 cache bandwidth.
Moreover, evolution of hardware architectures may require the adaption of all implemented combinations
or even a radical change of data structures.

A strategy to deal with the mentioned complexity is to choose a proper abstraction.
For example in linear algebra, a lot of operations may be implemented with only a
few simple inner-most kernels (e.g.\ \cite{VanZee2015}).
The advantage is two-fold, as firstly only a few kernels need to be optimised by an expert,
and secondly all projects that build up on these kernels may benefit from improvements.
Abstracting algorithms in terms of linear algebra operations, especially the
General Matrix-Matrix Multiplication (GEMM), is becoming popular for the numerical
solution of PDEs, too.
High-order discontinuous~Galerkin~(DG) or spectral element methods
usually consist of small element local operators, such that its implementations may be
expressed as sequence of small GEMMs \cite{Vincent2016, Uphoff2017, Hutchinson2016, Breuer2017}.
The matrices are usually small enough, such that data movement considerations within a GEMM become irrelevant.
At the same time, the overhead of calling BLAS and the choice of GEMM micro-kernels
becomes significant, such that specialised code generators for small GEMMs are
employed \cite{Heinecke2016b}, or even specialised code generators for small tensor
contractions are developed \cite{Breuer2017}.
An alternative ansatz is to express tensor operations in an appropriate intermediate
representation, which captures the underlying loop structure, and use compiler techniques
and cost models in order to generate efficient code \cite{Stock2011,Luporini2015,Kempf2018}.

The abstraction of tensor operations is not new but has been researched for decades
in the field of computational chemistry.
Here, one may distinguish two different levels of abstraction:
The first level of abstraction is the move from an expression involving several
tensors to binary tensor operations.
On this level of abstraction, algebraic transformations may be employed.
In the so-called strength reduction, for example, the optimal sequence
of tensor operations is determined (under a memory constraint) such that
the total number of floating point operations is minimised,
which is then usually magnitudes smaller than a naive implementation \cite{Lam1997}.
The most prominent representative is the Tensor Contraction Engine (TCE) \cite{Baumgartner2005},
but also a framework called Barracuda for GPUs exists \cite{Nelson2015}.
The second level of abstraction is the implementation of binary tensor operations,
especially for tensor contractions.
For the latter, a variety of algorithms exist, such as nested-loop code, Loop-over-GEMMs (LoG),
Transpose-Transpose-GEMM-Transpose (TTGT), and GEMM-like Tensor-Tensor
multiplication (GETT) \cite[and references therein]{Springer2018}.

In computational chemistry tensors may become very large, such that intermediate
results may not remain in memory but must be loaded from disk \cite{Baumgartner2005}.
So despite the fact that quantum chemistry models and high-order DG methods
may be abstracted as tensor operations, the situation is quite different for the latter,
as shall be outlined in the following:
\begin{itemize}[topsep=0pt]
 \item All matrices fit inside memory or even inside low-level caches.
       There is no need to optimise for memory usage or fuse loops to
       trade off memory usage and operation count \cite{Baumgartner2005}.
 \item Tensor operations do not need to be distributed among nodes.
 \item Data copies or index permutations may be amortised for large
       GEMMs \cite{Goto2008,Springer2018}, as these scale with $\mathcal{O}(N^2)$ compared
       to $\mathcal{O}(N^3)$, but this is usually not the case
       for small GEMMs \cite{Heinecke2016b}.
       (Which excludes TTGT.)
 \item All tensor dimensions and respective sparsity patterns in element-local
       operations are known a priori.
       Furthermore, spending a lot of time optimising individual tensor operations can
       be justified, as those are called millions of times during simulations \cite{Uphoff2017}.
 \item Sparsity can be dealt with explicitly, as it is feasible
       to compute the sparsity pattern of a tensor operation during
       the invocation of a code generator.
       There is no need to estimate the sparsity, as proposed in \cite{Lam1999}.
 \item Some optimisations are special to small GEMMs or DG methods, such
       as padding matrices for optimal vectorisation or tailored software prefetching
       schemes \cite{Heinecke2016a}.
\end{itemize}

Other publications and open-source packages focus on binary tensor contractions \cite{Solomonik2013,Springer2018,Shi2016,Li2015,Matthews2018},
GPUs \cite{Nelson2015}, or focus on loop transformations \cite{Stock2011,Luporini2015,Kempf2018},
where the latter lack support for sparse matrices in element-local operators and are to our
understanding not designed for use with code generators for small GEMMs.

We therefore present \textbf{Y}et \textbf{A}nother \textbf{Te}nsor \textbf{To}olbox (YATeTo),
which we explicitly designed for the use in PDE solvers.
In YATeTo, tensor operations are described in a high-level language
that is motivated by the Einstein convention.
An abstract syntax tree (AST) is formed on which algebraic transformations and other
optimisations are done.
Afterwards, code generators for small GEMMs are called and finally C++11-code is
generated which calls either generated code or calls BLAS libraries.
Our design does not prescribe the use of a specific implementation
of binary tensor contractions, but allows to choose the implementation based
on operation type and hardware architecture.
Note that our main focus are high-order DG and spectral element methods,
and that our ultimate goal is that the performance of the \yateto-generated code matches
the performance of hand-crafted code, but the toolbox itself may be used
for arbitrary applications where small tensor operations are required
and tensor dimensions and respective sparsity patterns are known a priori.

\section{Related work}
Many components of YATeTo build on previous algorithms and ideas.
In this section we summarise related work and necessary background.

\subsection{High-level language and representation}\label{sec:hll}
The need for high-level languages for tensor contractions was already
clear to A.~Einstein, who wrote in the foundations of general relativity:
\begin{quote}
 ``Daf\"ur f\"uhren wir die Vorschrift ein: Tritt ein Index in einem Term
 eines Ausdrucks zweimal auf, so ist \"uber ihn stets zu summieren,
 wenn nicht ausdr\"ucklich das Gegenteil bemerkt ist.'' \cite{Einstein1916}
\end{quote}
The so-called Einstein convention states that one shall sum over every index
that appears twice in a term.
Elegant domain-specific languages may be defined using this convention.
For example, in the Cyclops tensor framework one may write \cite{Solomonik2013}
\begin{codebox}
  \input{code/cyclops.pyg}
\end{codebox}
which is valid C++ (using operator overloading) and implements
\begin{displaymath}
z_{abij} = \sum_{mn} t_{abmn} \left(\sum_{ef}v_{mnef} t_{efij}\right).
\end{displaymath}
The Barracuda framework defines a similar high-level language \cite{Nelson2015}.

In the TCE, tensor operations are represented internally as function sequences
consisting of two types of formulae \cite{Lam1997}:
The first type is a multiplication formula of the form $f_r[\dots] = X[\dots] \times Y[\dots]$.
Depending on the indices inside the square brackets a different operation is represented.
E.g.\ $f_r[i,j] = X[i] \times Y[j]$ represents a tensor product whereas
$f_r[i] = X[i] \times Y[i]$ represents a Hadamard product.
(Note that two indices appear but are not summed over, that is,
this representation does not adhere to the Einstein convention.)
The second type is a summation formula of the form $f_r[\dots] = \sum_i X[\dots]$.
Here, the dimension of the left-hand side is always one less than the right-hand side,
and the dimension to be summed over is determined by the indices in the brackets.

\subsection{Strength reduction}

In many cases there are many mathematically equivalent ways to implement tensor
operations.
In~\cite{Baumgartner2005} the following motivating example is given,
where each dimension has size~$N$,
\begin{displaymath}
  S_{abij} = \sum_{cdefkl} A_{acik}B_{befl}C_{dfjk}D_{cdel} = 
	     \sum_{ck}\left(\sum_{df}\left(\sum_{el}B_{befl}D_{cdel}\right)C_{dfjk}\right)A_{acik}
\end{displaymath}
The first implementation corresponds to the naive implementation with 10 nested
loops, which requires $4N^{10}$ operations.
The second implementation saves intermediate results and only requires
$6N^6$ operations.
Hence, the second implementation is four orders of magnitudes cheaper than
the first implementation.

Naturally, one asks if and how optimal implementations can be found automatically.
A formal optimisation problem is set up in \cite{Lam1997} using the two types
of formulae mentioned in \Cref{sec:hll}.
It is shown that this optimisation problem is NP-complete
and an efficient search procedure is developed in \cite{Lam1997},
which exhaustively searches all valid formulae and hence finds the optimal
implementation.
As the number of dimensions arising in practice is typically small,
this search procedure is feasible.
The number of dimensions is also small enough in our application examples, see
\Cref{sec:applications}, hence there is no need to develop new
algorithms for strength reduction.

\subsection{Loop-over-GEMM (LoG)}\label{sec:logintro}
The authors of \cite{DiNapoli2014} give a summary of tensor contraction classes
and possible mappings to BLAS.
In a binary contraction, the contraction is classified in the number of free indices
of each input tensor. (That is, those indices which are not contracted.)
The most interesting case is when both tensors have at least one free index,
as in this case a mapping to GEMM is always possible,
if data is allowed to be copied or non-stride-one-GEMMs are considered.

Consider the following contraction from \cite{Springer2018}:
\begin{displaymath}
  C_{abcijk} := \sum_{m} A_{ijmc}B_{mkab}.
\end{displaymath}
Following \cite{Shi2016}, we have the following options in a LoG implementation:
We may batch modes yielding a non-free index $i$, denoted with $[i]$, we may flatten modes
yielding a single free index, denoted with $(ij)$, and we may transpose modes,
denoted with $A_{ij}^T$.
Note that transposes may only be applied to exactly 2 free modes.

A possible LoG implementation would be
\begin{equation}\label{eq:example_log_springer}
 C_{(ab)[c](ij)[k]} = \sum_{m} B_{m[k](ab)}^T A^T_{(ij)m[c]}.
\end{equation}
That is, we loop over the indices c and k.
Inside the loop, GEMM is called for subtensors of C, B, and A.
An implementation of \Cref{eq:example_log_springer} is listed in \cite[Listing~8]{Springer2018}.

\section{Tensors in discontinuous Galerkin methods}\label{sec:intro_dg}
The discontinuous Galerkin method has become a popular method
for solving hyperbolic PDEs.
The DG method has several advantages compared to other classic methods,
such as Finite Difference methods, Finite Volume methods, and continuous Finite Element methods:
It simultaneously allows to handle complex geometries, achieves
high-order convergence rates, and has an explicit semi-discrete form \cite{Hesthaven2008}.
Additionally, the high number of element-local operations and the small stencil
make it a promising candidate to exploit the capabilities of modern supercomputing hardware.
Before discussing the details of our tensor toolbox,
we want to motivate why tensor contractions are a natural abstraction for DG methods.

As an example, we consider the following linear system of PDEs:
\begin{equation}\label{eq:linearpde}
 \dfrac{\partial q_p}{\partial t} +
 A_{pq} \dfrac{\partial q_q}{\partial x} +
 B_{pq} \dfrac{\partial q_q}{\partial y} +
 C_{pq} \dfrac{\partial q_q}{\partial z} = 0.
\end{equation}
Multiplying the latter equation with a test function $\phi_k$, integrating over a domain $\Omega$,
and integrating by parts yields the corresponding weak form
\begin{multline}\label{eq:weakform}
 \int_{\Omega}\phi_k\dfrac{\partial q_p}{\partial t}\dd{V} +
 \int_{\partial \Omega}\phi_k\left(n_xA_{pq} + n_yB_{pq} +
    n_yC_{pq}\right)q_q\dd{S} \\
 + \int_{\Omega}\left(\dfrac{\partial \phi_k}{\partial x}A_{pq} q_q +
 \dfrac{\partial \phi_k}{\partial y}B_{pq} q_q +
 \dfrac{\partial \phi_k}{\partial z}C_{pq} q_q\right)\dd{V} = 0,
\end{multline}
where $(n_x, n_y, n_z)$ is the outward-pointing normal vector.
From here a sequence of steps follows to obtain a semi-discrete form, such as
partitioning the domain $\Omega$ in finite elements $\Omega_i$ (e.g.\ tetrahedra or hexahedra),
introducing a numerical flux to weakly couple the finite elements,
transforming the physical elements to a reference element,
and introducing a polynomial approximation of the quantities on the reference element \cite{Atkins1998,Dumbser2006,Hesthaven2008}.
After these steps, the semi-discrete form is a system of ODEs, which
may, for example, be solved with a Runge-Kutta scheme or the ADER approach \cite{Dumbser2006}.

We skip some steps for simplicity, and assume that the quantities are discretised as
$q_p~=~Q_{lp}~\phi_l(x,y,z)$ at a given point in time on an element $\Omega_i$.
Here, $\phi_l(x,y,z)$ are a set of $\mathcal{B}$ polynomial basis functions and $Q$ is a $\mathcal{B}\times P$
matrix, storing the coefficients of the basis expansion for every quantity, where
$P$ is the number of quantities.
By inserting our discretisation we may write the second line of \eqref{eq:weakform}
in the following way
\begin{equation}\label{eq:stiffnesskernel}
 Q_{lq} A_{pq}\int_{\Omega_i}\dfrac{\partial \phi_k}{\partial x}\phi_l\dd{V} +
 Q_{lq} B_{pq}\int_{\Omega_i}\dfrac{\partial \phi_k}{\partial y}\phi_l\dd{V} +
 Q_{lq} C_{pq}\int_{\Omega_i}\dfrac{\partial \phi_k}{\partial z}\phi_l\dd{V}.
\end{equation}
The integrals in \eqref{eq:stiffnesskernel} may be pre-computed and stored as (so-called) stiffness matrices
$K^{x}$, $K^{y}$, and $K^{z}$ of size $\mathcal{B}\times \mathcal{B}$.\footnote{Note that these matrices need to be only computed for the reference element
in an actual implementation.}
The implementation in terms of GEMM is then given by the sequence of matrix chain products
$K^xQA^T+K^yQB^T+K^zQC^T$.

When one uses the unit cube as reference element, then it is possible to use a spectral basis,
where the basis functions and test function are given by
$\phi_{k(x,y,z)} = \psi_x(x)\psi_y(y)\psi_z(z)$.
The degrees of freedom are then stored in a 4-dimensional
tensor $Q$, such that $q_p = Q_{lmnp} \psi_l(x) \psi_m(y) \psi_n(z)$.
In each dimension we have $N+1$ basis functions, where $N$ is the polynomial degree, and $Q$ has $(N+1)^3 P$ entries.
Inserting the discretisation and the test functions in the second line of \eqref{eq:weakform}, we obtain
\begin{displaymath}
 Q_{lmnq} A_{pq}K_{xl}M_{ym}M_{zn} +
 Q_{lmnq} A_{pq}M_{xl}K_{ym}M_{zn} +
 Q_{lmnq} A_{pq}M_{xl}M_{ym}K_{zn},
\end{displaymath}
where $K_{ij} = \int_{0}^1\frac{\partial \psi_i}{\partial x}\psi_j\dd{x}$
and $M_{ij} = \int_{0}^1\psi_i\psi_j\dd{x}$.
We note here that the index order \textit{lmnp} is chosen
arbitrarily.
In fact, the index permutations \textit{plmn}, \textit{lpmn},
and \textit{lmpn} are equally viable, and it depends on the 
application which is best.
E.g.\ if $P$ is large and $N+1$ is smaller than the SIMD vector
width then it might be beneficial to store index
$p$ continuous in memory.
Otherwise, if $N+1$ is large and $P$ is small then
either of \textit{lmn} should be continuous in memory.

As a final example, one may solve multiple problems concurrently.
E.g.\ one may add another dimension to the degrees of freedom $Q_{lp}$ in
\eqref{eq:stiffnesskernel} to obtain a 3D tensor $Q_{slp}$, which stores
multiple degrees of freedom \cite{Breuer2017}.
The matrix chain product $K^xQA^T$ then becomes the tensor contraction sequence
$K^x_{kl}Q_{slq}A_{pq}$ and may be implemented using the Loop-over-GEMM approach,
e.g.\ $(Q_{(sl)q}A_{pq}^T)_{sl[p]}(K^x_{kl})^T$.

\section{Yet Another Tensor Toolbox}

\begin{figure}
 \includegraphics[width=.95\textwidth]{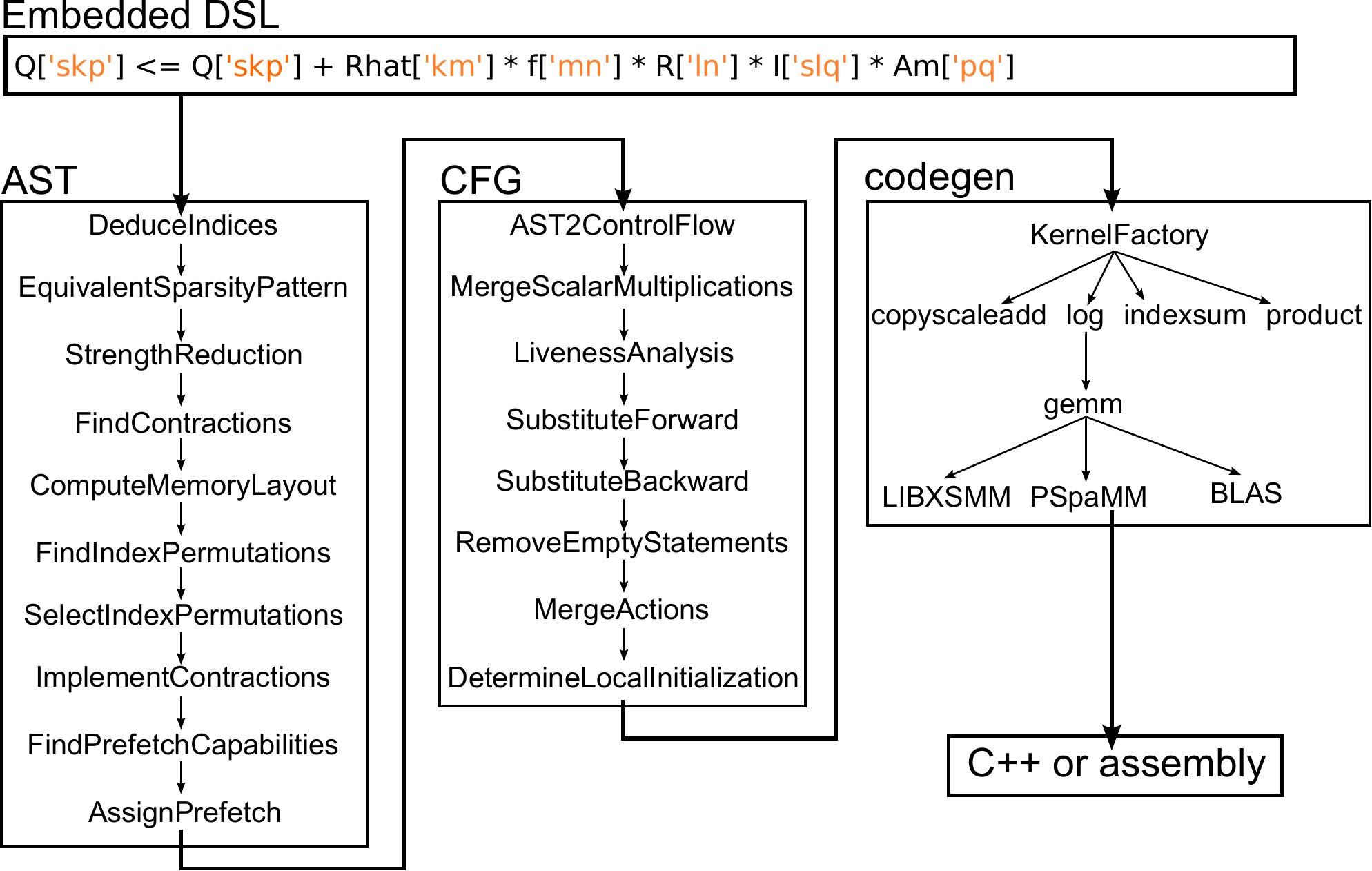}
 \caption{Overview of \yateto, showing the path from the high-level language
	  to the generated code. The first stage (left box) operates on an
	  abstract syntax tree, in the second stage (middle box) the representation
	  is changed to a simple control flow graph without branches, and
	  in the last stage (right box) different back-ends are called
	  which output C++ or assembly.}
  \label{fig:overview}
\end{figure}

The input to our tensor toolbox is a domain-specific language (DSL),
which naturally resembles the Einstein notation.
An abstract syntax tree (AST) is derived from an expression.
Subsequently, the tree is shaped using a sequence of visitors \cite{Gamma1995}.
Afterwards, we transform the AST to a control flow graph (CFG) for
standard compiler techniques \cite{Seidl2012}.
Note that the CFG is very simple as the DSL does not feature loops and branches.
Finally, a code generator is called which generates either C++-11 code
or may make use of existing code generators or BLAS libraries.
An overview is shown in \Cref{fig:overview}.
\yateto~is open-source and available on \url{www.github.com/SeisSol/yateto}.

In the following, we are going to introduce our DSL and elaborate algorithms and
design choices of each of the three stages.

\subsection{High-level description of tensor operation}

The DSL we created is embedded into Python 3.
An embedded DSL has the advantage that we do not need our own parser and lexer,
which decreases the complexity of the software.
Furthermore, a user may use any feature of Python to form expressions,
for example including loops, branches, classes, and lambda expressions.

The basic data types are the classes Tensor and Scalar.
A tensor definition includes a name and a shape tuple, e.g.\ as in the following:
\begin{codebox}
  \input{code/basic_tensor.pyg}
\end{codebox}
A and C are matrices, B a third order tensor, and w is a vector, where all
dimensions are of size N.

From a Tensor, an object of type IndexedTensor may be derived by supplying an index string.
The convention is that the length of the index string must be equal to the number
of dimensions of the tensor and that only lower- and upper-case letters may be used.\footnote{Note
that this limits the number of distinct indices that may appear in an expression to 52.
We think that this is not an issue in practice, hence we favour simplicity
instead of generality.}
Objects of type IndexedTensor may appear in expressions, e.g.\

\begin{codebox}
  \input{code/kernel_example.pyg}
\end{codebox}

The design is inspired by \cite{Solomonik2013} (see \Cref{sec:hll}),
and represents the operation
\begin{displaymath}
C_{ij} := 2\cdot C_{ij} + \sum_{l=0}^{N-1}\sum_{k=0}^{N-1} A_{lj} B_{ikl} w_k
\quad\text{or equivalently}\quad C := 2\cdot C + (B\times_2 w) A.
\end{displaymath}
The variable kernel contains an AST with nodes
Assign, Add, ScalarMultiplication, and Einsum, where the latter
represents a tensor operation using the Einstein convention.

The AST is built by overloading the binary operators '*' and '+'.
One of the following four actions is taken whenever one of the operators is
invoked for operands op1 and op2, where nodeType = Einsum for '*' and
nodeType = Add for '+':
\begin{enumerate}
 \item type(op1) = nodeType and type(op2) = nodeType: Append op2's children to op1.
 \item type(op1) = nodeType and type(op2) $\neq$ nodeType: Append op2 to op1.
 \item type(op1) $\neq$ nodeType and type(op2) = nodeType: Prepend op1 in op2.
 \item type(op1) $\neq$ nodeType and type(op2) $\neq$ nodeType: Return nodeType of op1 and op2.
\end{enumerate}
These actions ensure that a node of type nodeType does not have children
of type nodeType, but all tensors w.r.t.\ the same nodeType are merged.
This is important, especially for Einsum and the algorithms presented in \Cref{sec:eqspp}
and \Cref{sec:strengthred}, which require all tensors participating in a
tensor operation to be known.
Note that we do not apply distributive laws, e.g.\ in the expression
$\sum_k A_{ik} (B_{kj} + C_{kj})$ the addition is always evaluated before
the Einstein sum.

For scalar multiplications we take similar actions to ensure that
ScalarMultiplication is always the parent node of an Einsum.

\subsection{Equivalent sparsity patterns}\label{sec:eqspp}
\begin{figure}
  \begin{subfigure}{\textwidth}
    \centering
    \includegraphics[width=\textwidth]{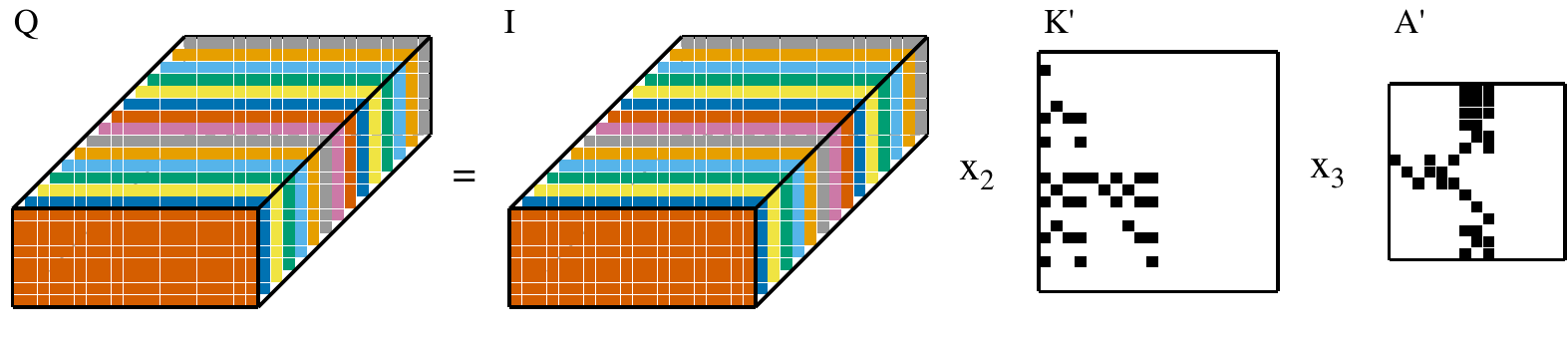}
    \caption{Original sparsity patterns, stored in memory.}\label{fig:volume_originalspp}
  \end{subfigure}
  \begin{subfigure}{\textwidth}
    \centering
    \includegraphics[width=\textwidth]{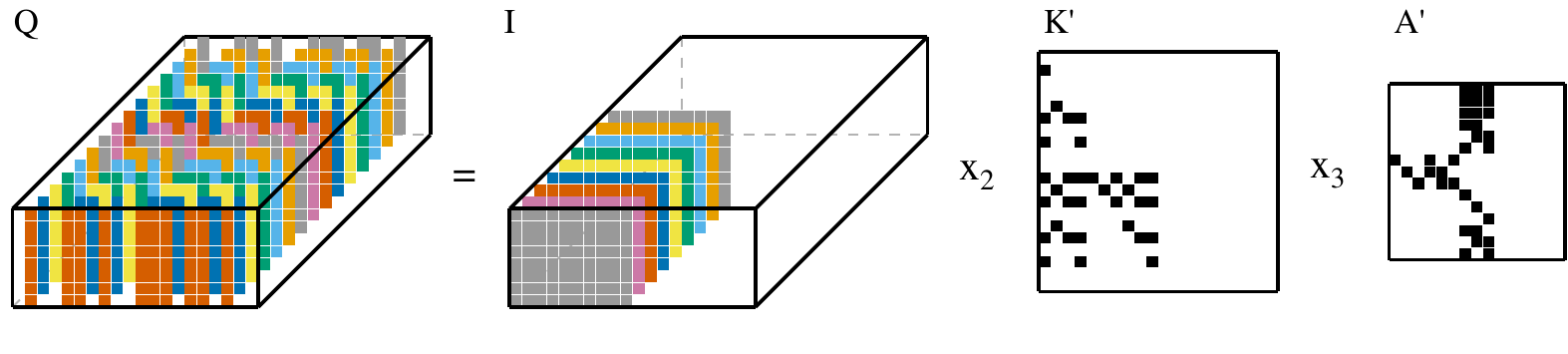}
    \caption{Equivalent sparsity patterns, used during contractions.}\label{fig:volume_equivalentspp}
  \end{subfigure}
  \caption{We show the kernel $Q_{skp} = I_{slq}K_{lk}A_{qp}$ (depicted with mode-n product,
  i.e.\ $Q=I\times_2 K^T\times_3A^T$), which is similar to an application example in \cite{Uphoff2016}
  but with an additional dimension added to the degrees of freedom.
  The degrees of freedom (I, Q) are given as full tensor (\Cref{fig:volume_originalspp}),
  but we detect that large blocks in I do not influence the final result,
  as they are multiplied with zeros in the tensor contractions (\Cref{fig:volume_equivalentspp}).
  }\label{fig:eqsppdemonstration}
\end{figure}

The notion of equivalent sparsity patterns (EQSPPs) and an algorithm to compute them
was introduced in \cite{Uphoff2016} for matrix chain products (MCP).
It is defined as the minimal sparsity patterns of the involved matrices
that leaves the result of an MCP unchanged.
In the latter definition, no cancellation is assumed, which means
that adding non-zero entries never results in zero \cite{Cohen1998}.
First, the assumption allows to work with boolean matrices alone.
Second, we assume that the sparsity pattern is known but the values are unknown
(or known only at run-time), which has the following implication:
In an inner product, which is the core of matrix multiplication,
we can always find a pair of non-orthogonal vectors.
That is, if the values of at least one matrix are unknown,
we cannot guarantee cancellation.

The concept of EQSPPs can be illustrated with the following example
(where the entries may be either scalars or dense block matrices):
$$
\begin{pmatrix}
 K_1 & 0
\end{pmatrix}
\begin{pmatrix}
 Q_{11} & Q_{12} \\
 Q_{21} & Q_{22} \\
\end{pmatrix}
\begin{pmatrix}
 A_1 \\ 0
\end{pmatrix} =
K_1Q_{11}A_1
$$
Here, the multiplication with the first matrix removes
$Q_{21}$ and $Q_{22}$ from the right-hand side and the multiplication
with the third matrix removes $Q_{12}$ and $Q_{22}$ from the right-hand side.
Thus, the equivalent sparsity pattern of the second matrix is only non-zero
for the top-left block ($Q_{11}$).

In \cite{Uphoff2016} an algorithm to compute EQSPPs for MCPs is presented,
which is based on a graph based representation.
The extension of the MCP graph to general tensor operations is not straightforward,
thus we derive an algorithm that is not based on a graph
in the remainder of this section.

\subsubsection{Formalisation of tensor operations}
So far we have defined tensor operations via handwaving and examples,
but here we need to introduce a formal definition of the
``\texttt{*}''-expression in our DSL.

First of all, we formalise the labelling of the dimensions of a tensor $T$.
It consists of a set of indices and an assignment of indices to dimensions.
We denote the set of indices with calligraphic letters,
e.g.\ $\mathcal{T} \subset \mathscr{P}(\Sigma)$, which must be a subset of the
power-set of an alphabet $\Sigma$, where e.g.\ $\Sigma=\{a,\dots,z\}$.
We require $|\mathcal{T}|=\dim(T)$.\footnote{Note that this requirement excludes
the trace of a matrix $T_{ii}$ or similar expressions.
However, one may write the trace as $T_{ij}\delta_{ij}$,
which conforms with the requirement.}
The assignment of indices to dimensions is a bijective mapping
$P_T : \mathcal{T} \rightarrow \{1,\dots,\dim(T)\}$.
The range of an index $i$ is always $I_{P_T(i)}:=\{0,\dots,\size(i)-1\}$,
where $\size(i)\in\mathbb{N}$.
\begin{ex}
 Let $A\in\mathbb{R}^{2\times 2\times 2}$.
 The index set of \texttt{A['jik']} is $\mathcal{A}=\{\mathtt{i},\mathtt{j},\mathtt{k}\}$, the map $P_A$
 is given by $P_A(\mathtt{j}) = 1, P_A(\mathtt{i}) = 2, P_A(\mathtt{k}) = 3$, and $\size(\mathtt{i})=\size(\mathtt{j})=\size(\mathtt{k})=2$.
\end{ex}
Second, assume a ``\texttt{*}''-expression is formed by $n$ tensors $T^{1},\dots,T^{n}$,
and the result tensor $U$.
Einstein's convention requires us to sum over all indices that appear
twice, which we formalise in the following:
The index space of tensor $T^k$ is $\mathbb{S}(T^k) = I_1\times \dots\times I_{\dim(T^k)}$,
where $I_d\subset\mathbb{N}_0$ is the index range of the $d$-th dimension.
In order to span an iteration space, we introduce 
the global index set $\mathcal{G} = \bigcup_{k}\mathcal{T}^k$,
and the global index space
$\mathbb{G} = \bigtimes_{k=1}^{N} \{0,\dots,\size(P_G^{-1}(k))-1\}$,
where $N=|\mathcal{G}|$.
Similar to a tensor, we assign global indices to dimensions
of the global index space with the bijective function
$P_G : \mathcal{G} \rightarrow \{1,\dots,N\}$.\footnote{The function $P_G$
may for example be obtained by ordering the global index set lexically.}
In order to ``access'' entries of a tensor, we introduce projection and permutation functions
$\pi_{T^k} : \mathbb{G} \rightarrow \mathbb{S}(T^k)$, where
$$
\pi_{T^k}(i_1,\dots,i_N) = \left(i_{P_G\left(P_{T^k}^{-1}(d)\right)} : 1 \leq d \leq |\mathcal{T}^k|\right).
$$
Entries of $T^k$ are accessed with the notation $T^k[{\pi_{T^k}(i)}]$, where $i\in\mathbb{G}$,
or shortly with $T^k_{\tp{i}}$, where the permutation and
projection function is understood from context.
We define the restriction of the global index set to an index
$i\in\mathbb{S}(T)$ as $R_T(i) = \{j\in\mathbb{G} : \pi_T(j) = i\}$.

With the set up formalism, we are able to precisely state the meaning
of a ``\texttt{*}''-expression:
\begin{equation}\label{eq:gentensorop}
 \forall i \in \mathbb{S}(U): U_{i} = \sum_{j\in R_U(i)} T^{1}_{\tp{j}}\cdot\ldots\cdot T^{n}_{\tp{j}}.
\end{equation}
\begin{ex}
 Consider tensors $A\in\mathbb{R}^{2\times 2\times 2}$ and $B\in\mathbb{R}^{2\times 2}$.
 In Einstein notation, the mode-2 product is written as $C_{ijk} = A_{ilk}B_{jl}$
 and in our DSL as \texttt{C['ijk'] <= A['ilk'] * B['jl']}.
 
 In our formal notation, the global index space is
 $\mathbb{G} = I\times J\times K\times L$, where $I=J=K=L=\{0,1\}$.
 The permutation and projection functions are $\pi_A:(i,j,k,l) \mapsto (i,l,k)$,
 $\pi_B:(i,j,k,l) \mapsto (j,l)$, and $\pi_C:(i,j,k,l) \mapsto (i,j,k)$.
 We may compute the entries of C with
 $$
  (i,j,k)\in I\times J\times K: C_{(i,j,k)} = \sum_{(i',j',k',l') \in \{i\}\times\{j\}\times \{k\}\times L} A_{(i',l',k')} B_{(j',l')}.
 $$
\end{ex}

\subsubsection{Computation of EQSPPs for tensor operations}
As we assume no cancellation in sums, it is sufficient to model
a sparsity pattern as a boolean tensor.
The EQSPP of a tensor $T^k$ is called $\hat{\Theta}^k$, where
$\hat{\Theta}^{k}_q \in \{0,1\},\; q\in\mathbb{S}(T^k)$.
The ``hat''-tensor $\hat{T}^k$ is the original tensor $T^k$ masked
with $\hat{\Theta}^{k}$.
In operations involving $\hat{\Theta}^{k}$ we identify $+$ with $\vee$ and $\cdot$
with $\wedge$.

Having set up the necessary formalism, we are ready to formally define EQSPPs:
\begin{definition}
 Assuming no cancellation, we call $\hat{\Theta}^{k}$
 equivalent sparsity patterns w.r.t.\ \eqref{eq:gentensorop} if
 \begin{enumerate}
  \item $U = \hat{U}$, where $\forall i\in\mathbb{S}(U): \hat{U}_{i} = \sum_{j\in R_U(i)} \hat{T}^{1}_{\tp{j}}\cdot\ldots\cdot \hat{T}^{n}_{\tp{j}}$
  and $\hat{T}^k_q = \left\{\begin{array}{ccc}
                       T^k_q & \text{ if } & \hat{\Theta}_q^k = 1 \\
                       0 & \text{ if } & \hat{\Theta}_q^k = 0. \\
                      \end{array}\right.$,
  \item The number of non-zeros of $\hat{\Theta}^{k}$, $k=1,\dots,n$, is minimal, that is,
	we cannot set a non-zero to zero without implying $\exists i\in\mathbb{S}(U):U_i \neq \hat{U}_i$.
 \end{enumerate}
\end{definition}
As we assume no cancellation, we may disregard the sums over indices,
and instead need to only consider a product over all tensors
on the right-hand side:
\begin{lemma}\label{lemma:outerproduct}
  The EQSPPs w.r.t.\ \eqref{eq:gentensorop} are equivalent to
  the EQSPPs w.r.t.\
  \begin{displaymath}
  \forall j\in\mathbb{G}: Z_j := T^{1}_{\tp{j}}\cdot\ldots\cdot T^{n}_{\tp{j}}.
  \end{displaymath}
\end{lemma}
\begin{proof}
 Clearly, $\forall i\in\mathbb{S}(U): U_i = \sum_{j\in R_U(i)} Z_{j} = \sum_{j\in R_U(i)} \hat{Z}_{j}$, hence condition 1 is fulfilled.
 
 In order to check condition 2, assume there is another set of EQSPPs, $\widebar{\Theta}^k$,
 which has less non-zeros than $\hat{\Theta}^k$.
 Then $\exists f\in\mathbb{G}: Z_{f} = \hat{Z}_{f}\neq\widebar{Z}_{f}$
 (otherwise $\hat{T}^k$ would not be minimal).
 As $\bigcup_{i\in\mathbb{S}(U)} R_U(i) = \mathbb{G}$ and $Z_{f}$ does not get cancelled in sums,
 it follows that there exists an index $g\in\mathbb{S}(U)$
 where $\widebar{U}_g \neq U_g$. 
\end{proof}

Our main result is that we may cast the computation of EQSPPs
as boolean tensor operations on the original sparsity patterns:
\begin{prop}
  The EQSPPs w.r.t.\ \eqref{eq:gentensorop} are given by
  \begin{equation}\label{eq:compeqspp}
  \forall q \in \mathbb{S}(T^k): \hat{\Theta}^{k}_{q} := \sum_{l\in R_{T^k}(q)} \Theta^{1}_{\tp{l}}\cdot\ldots\cdot\Theta^{n}_{\tp{l}},
  \end{equation}
  where $\Theta^l$ is the sparsity pattern of $T^l$, $l=1,\dots,n$.
\end{prop}
\begin{proof}
We are going to show that \eqref{eq:compeqspp} computes the EQSPPs
for the product $Z$ of tensors $T^1,\dots,T^n$
(see \Cref{lemma:outerproduct}).
In order to satisfy condition 1, we only need to show that the sparsity
patterns of $Z$ and $\hat{Z}$ are identical.
The sparsity pattern of $Z$ is given by $\zeta$, where
$
 \forall j \in \mathbb{G}: \zeta_j := \Theta^{1}_{\tp{j}}\cdot\ldots\cdot\Theta^{n}_{\tp{j}},
$
and the sparsity pattern of $\hat{Z}$ is given by $\hat{\zeta}$, where
$
 \forall j \in \mathbb{G}: \hat{\zeta}_j := \hat{\Theta}^{1}_{\tp{j}}\cdot\ldots\cdot\hat{\Theta}^{n}_{\tp{j}}.
$

First, we show that
$\forall j\in\mathbb{G}: \hat{\zeta}_j = \zeta_j$.
Note that
\begin{equation}\label{eq:producttricks}
  \hat{\Theta}^{k}_{q} =
  \sum_{l\in R_{T^k}(q)} \Theta^{1}_{\tp{l}}\cdot\ldots\cdot\Theta^{n}_{\tp{l}} =
  \sum_{l\in R_{T^k}(q)} \zeta_l =
  \Theta_q^k\cdot \sum_{l\in R_{T^k}(q)} \zeta_l,
\end{equation}
where the latter equality follows from idempotence and
$\Theta_\tp{l}^k=\Theta_q^k$ if $l\in R_{T^k}(q)$ by the definition of $R_{T^k}$.
Using \eqref{eq:producttricks} we obtain
\begin{equation}\label{eq:zetaeqzetahat}
\hat{\zeta}_j =
\hat{\Theta}^{1}_{\tp{j}}\cdot\ldots\cdot\hat{\Theta}^{n}_{\tp{j}} =
\zeta_{j}\cdot
\left(\sum_{l\in R_{T^1}(\pi_{T^1}(j))} \zeta_{l}\right)\cdot\ldots\cdot
\left(\sum_{l\in R_{T^n}(\pi_{T^n}(j))} \zeta_{l}\right) =
\zeta_{j},
\end{equation}
where the last identity follows from $j\in R_{T^k}(\pi_{T^k}(j))$ and the absorption law.
\Cref{eq:zetaeqzetahat} is true for all $j\in\mathbb{G}$, hence $\zeta=\hat{\zeta}$.

Second, in order to satisfy condition 2, we show that the number of non-zeros is minimal:
Assume that there exist another set of sparsity patterns $\widebar{\Theta}^k$,
with less non-zeros than $\hat{\Theta}^k$.
Then, $\exists k\in\{1,\dots,n\}:\exists q\in\mathbb{S}(T^k):
\widebar{\Theta}_q^k=0 \wedge \hat{\Theta}_q^k=1$.
Clearly, $\Theta_q^k=1$ and $\sum_{l\in R_{T^k}(q)} \zeta_l = 1$, because
otherwise $\hat{\Theta}_q^k\neq 1$.
Hence, there must be at least one $p\in R_{T^k}(q)$ such that $\zeta_p=1$.
As $\pi_{T^k}(p)=q$ it immediately follows that
$$
\widebar{\zeta}_p :=
\widebar{\Theta}^1_\tp{p}
\cdot\ldots\cdot\widebar{\Theta}^{k-1}_\tp{p}
\cdot\widebar{\Theta}^k_q
\cdot\widebar{\Theta}^{k+1}_\tp{p}
\cdot\ldots\cdot\widebar{\Theta}^n_\tp{p} = 0
\neq \zeta_p = 1,
$$
which violates condition 1.
\end{proof}

\subsubsection{Implementation and discussion}
We have applied our EQSPP algorithm to an application example
from \cite{Uphoff2016}, where we extended the degrees of freedom
by an additional dimension (cf.~\Cref{sec:appseissol}).
The original SPPs and EQSPPs can be seen in \Cref{fig:eqsppdemonstration}.
We observe that the non-zero entries of $K$ and $A$ induce
additional zeros in $I$ which may be disregarded when computing
the tensor operation.
For example, in the evaluation of $I_{slq}K_{lk}$ we may restrict the
loop ranges of $q$ and $l$.

The implementation of EQSPP-computation with \eqref{eq:compeqspp} is straightforward,
as we require the same kind of tensor operations that we support within
\yateto, the only difference being that the tensors are boolean.
E.g.\ we may apply strength reduction (cf.~\Cref{sec:strengthred})
in order to reduce the amount of computations.
Still, our method to compute EQSPPs is more expensive to compute than
to simply evaluate the original tensor operation.
Nevertheless, the cost of computing EQSPPs is negligible in comparison
to the possibly millions of times a kernel is called within a DG scheme.

\subsection{Strength reduction}\label{sec:strengthred}
We already mentioned in \Cref{sec:strengthred} that finding the sequence
of tensor operations with minimal operation count is an NP-hard problem,
but an efficient enumeration procedure exists for tensor dimensions
appearing in practice \cite{Lam1997}.
The same enumeration procedure may also be used when sparse tensors are
involved, but in this case the sparsity patterns of intermediate results
are required, or at least an estimate of the sparsity \cite{Lam1999b}.

In \yateto~we assume that tensors are small enough, such that it is feasible
to explicitly compute all intermediate products during strength reduction.
The number of operations is then determined in the following way \cite{Lam1999b}:
For a multiplication formula $V[\dots] = X[\dots] \times Y[\dots]$
the number of operations is equal to the number of non-zeros in $V$.
For a summation formula $W[\dots] = \sum_i Z[\dots]$ the number of
operations is equal to the number of non-zeros in $Z$ minus the number
of non-zeros in $W$.

An example of an intermediate AST after strength reduction is shown 
in \Cref{fig:strengthred_ast}.

\begin{figure}
 \begin{subfigure}[b]{.3\linewidth}
  \centering
  \begin{tikzpicture}[sibling distance=3em,
  every node/.style = {align=center}]]
  \node {$T^{(1)}_{ij}:\text{Einsum}$}
      child { node {$A_{lj}$} }
      child { node {$B_{ikl}$} }
      child { node {$w_{k}$} };
\end{tikzpicture}
  \caption{Initial AST from DSL.}
  \label{fig:einsum_ast}
 \end{subfigure}
 \begin{subfigure}[b]{.3\linewidth}
  \centering
  \begin{tikzpicture}[sibling distance=4em, level distance=3em,
  every node/.style = {align=center}]]
  \node {$T^{(4)}_{ij}:\sum_l$}
    child { node {$T^{(3)}_{lji} : \times$}
      child { node {$A_{lj}$} }
      child { node {$T^{(2)}_{il}:\sum_k$}
	child { node {$T^{(1)}_{ikl}:\times$}
	  child { node {$B_{ikl}$} }
	  child { node {$w_{k}$} }
	}
      }
    };
\end{tikzpicture}
  \caption{AST after strength reduction.}
  \label{fig:strengthred_ast}
 \end{subfigure}
 \begin{subfigure}[b]{.3\linewidth}
  \centering
  \begin{tikzpicture}[sibling distance=4em, level distance=4em,
  every node/.style = {align=center}]]
  \node {$T^{(2)}_{ij}: \text{LoG}$}
    child { node {$T^{(1)}_{il}: \text{LoG}$}
      child { node {$B_{ikl}$} }
      child { node {$w_{k}$} }
    }
    child { node {$A_{lj}$} };
\end{tikzpicture}
  \caption{Final AST.}
  \label{fig:final_ast}
 \end{subfigure}
 \caption{Overview of major stages during AST transformation of an Einsum node (left).
 After determining the equivalent sparsity patterns, the Einsum node is transformed
 during strength reduction, yielding an operation-minimal tree (middle).
 This tree is binary and consists of Product and IndexSum nodes only.
 Finally, a mapping to Loop-over-GEMM is found which minimises the cost function
 described in \Cref{sec:optindperm} (right).}
\end{figure}
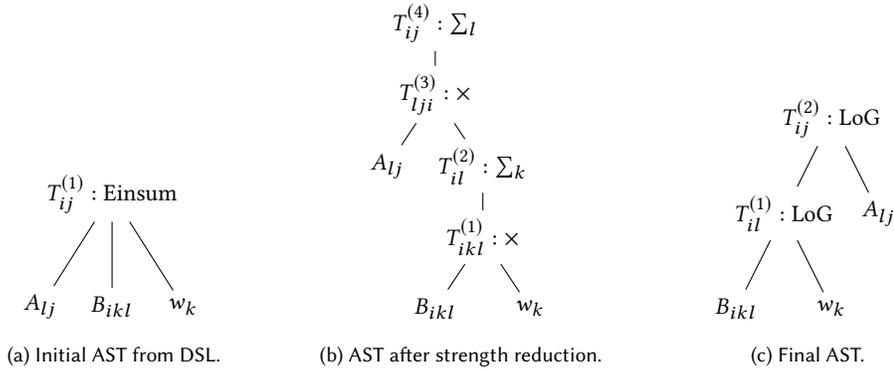

\subsection{Memory layouts}
Following \Cref{fig:overview} the next step is to compute memory layouts of tensors.
The memory layout influences the possibility to fuse indices in a
LoG-implementation.
Thus the layout influences the cost function in \Cref{sec:optindperm} and 
needs to be fixed at this point already.
In the following we give an overview over the two classes of memory layouts
that are currently supported.

\subsubsection{Dense layout}
Let $A\in\R^{n_1\times n_2\times \dots \times n_d}$.
We call the tuple $(n_1,n_2,\dots,n_d)$ the shape of a tensor.
A simple memory layout of a dense tensor is the
``column-major'' storage:
$$
A_{i_1\dots i_d} := A\left[\sum_{k=1}^d i_k s_k\right],
$$
where the so-called stride is given by $s_k=\prod_{l=0}^{k-1}n_l$ with $n_0:=1$.
That is, the tensor is stored linearly in a 1D array, such that
the first dimension varies fastest and the last dimension varies slowest in memory.
The stride $s_k$ gives us the number of floating numbers we have to skip when we 
increase index $i_k$ by one.

From \Cref{sec:eqspp} we expect that we also have tensors with large zero blocks,
where it would be wasteful to store the full tensor.
Hence, our next evolution is the bounding-box column-major storage:
$$
A_{i_1\dots i_d} := A\left[\sum_{k=1}^d (i_k-b_k) t_k\right],
$$
where $t_k=\prod_{l=0}^{k-1}(B_l-b_l)$ and $0 \leq b_k < B_k \leq n_k$.
This memory layout models the case where $A$ is non-zero within the 
index set $[b_1,B_1)\times\dots\times [b_d,B_d)$ and zero otherwise.

Finally, it might be beneficial to align the fastest dimension.
That is, the number of bytes in a fibre in the first dimension should
be divisible by the architecture's SIMD vector length.
In some cases, one has to add artificial zeros to the first dimension,
hence we allow that the bounding interval of the first
dimension is larger than the size of the first dimension.
Given an initial (minimal) bounding interval $[b_1, B_1)$ we
align the bounding interval as in \cite{Uphoff2016}, that is,
$$
[b_1', B_1') := [b_1 - b_1 \bmod v, B_1 + (v - B_1 \bmod v) \bmod v),
$$
where $v$ is the number of floating point values that fit into
a SIMD vector (e.g.\ using AVX-512 one sets $v=8$ for double precision
or $v=16$ for single precision).

We conclude the presentation of dense layouts with a possible pitfall
of using bounding boxes and alignment.
It might be beneficial to fuse multiple indices when mapping a tensor
contraction to a Loop-over-GEMM (see \Cref{sec:logintro}).
Fusing indices $i_a\dots i_b$ is always possible when \cite{Shi2016}
\begin{equation}\label{eq:compfuseindices}
\forall i \in [a,b): t_{i+1} = n_it_i.
\end{equation}
Other cases would require more involved compatibility conditions.
For example in the GEMM $C_{(ij)(kl)} = A_{(ij)m}B_{m(kl)}$, where
artificial zeros are added for dimension $i$ in $A$ and dimension $k$
in $B$, then one may only fuse $(ij)$ and $(kl)$ whenever
$C$ has the same number of artificial zeros in dimension $i$ and $k$.
Otherwise one needs temporary storage and an additional data copy.
In order to avoid possible data copies and complications arising
in subsequent optimisation steps we only allow fused indices
when \eqref{eq:compfuseindices} is fulfilled.
Conversely, aligned layouts or bounding box layouts have to be
considered carefully as they may prohibit fusing indices, which
leads to smaller (and possibly slower) GEMMs.

\subsubsection{Compressed sparse column (CSC) layout}
Sparse matrices may appear in discontinuous Galerkin methods,
e.g.\ the stiffness matrices are sparse when an orthogonal set of
basis functions is used.
We have a limited support for the CSC format:
CSC matrices may appear in GEMM or Loop-over-GEMM calls but only
in a sparse x dense or a dense x sparse GEMM.

\subsubsection{Other (sparse) formats}
It might be surprising that we describe Equivalent Sparsity Patterns
in great detail but currently do not offer a genuine sparse tensor
memory layout.
First of all, we want to point out that if an EQSPP induces a large zero
block, then we are able to exploit this block using the bounding box
memory layout or we may let a LoG operate only on a subtensor.
Second, to the authors' knowledge, an efficient code generator or efficient
routines for general small sparse tensors do currently not exist.

\subsection{Optimal index permutations}\label{sec:optindperm}
\begin{figure}
 \begin{tikzpicture}[level distance=3em,
  every node/.style = {align=center},
  level 1/.style={sibling distance=14em},
  level 2/.style={sibling distance=9em}, 
  level 3/.style={sibling distance=3em}]
  \node {$Q_{skp}:$ Contraction}
    child { node {$T^{(2)}_{pns} :$ Contraction}
      child { node {$A_{pq}$} }
      child { node {$T^{(1)}_{qns}:$ Contraction}
	child { node {$R_{ln}$} }
	child { node {$I_{slq}$} }
      }
    }
    child { node {$T^{(3)}_{kn} :$ Contraction}
      child { node {$f_{mn}$} }
      child { node {$\hat{R}_{mk}$} }
    };
\end{tikzpicture}
 \caption{A possible intermediate AST of $Q_{skp}:=\hat{R}_{km}f_{mn}R_{ln}I_{slq}A_{pq}$.
 We have the freedom to choose the index permutations of the temporary tensors
 $T^{(1)}$, $T^{(2)}$, and $T^{(3)}$.}
  \label{fig:permutation_tree}
\end{figure}
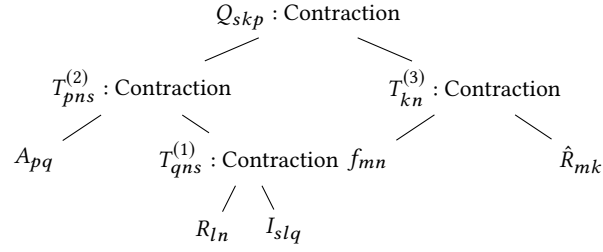
The strength reduction step converts each Einsum node to a binary tree
consisting only of IndexSum and Product nodes.
In a first step, we identify contractions in the tree,
which are (multiple) IndexSum nodes followed by a single Product node.
A sub-tree that was identified to be a contraction is replaced
by a new node of type Contraction.
For example, we might end up with the (sub-)tree shown in \Cref{fig:permutation_tree}.

Each Contraction node shall now be mapped to an implementation as Loop-over-GEMM.
In general, the mapping to GEMM is not unique, but several implementations are
possible (e.g.\ \cite[Table~II]{Shi2016}).
A systematic way to list all possible implementations is given in \cite[Listing~7]{Springer2018}.
As accurately predicting the performance of each implementation is difficult,
heuristics may be employed to select the most promising candidate,
e.g.\ one may choose the variant with the largest GEMM~\cite{Springer2018}.

The algorithm to list all LoG-implementations in \cite{Springer2018} assumes
that the order of all indices are known.
In our case we have temporary tensors in which an order of
indices is not prescribed.
(E.g.\ $T^{(1)}$, $T^{(2)}$, and $T^{(3)}$ in \Cref{fig:permutation_tree}.)
The index order has to be chosen carefully as it decides possible mappings
to GEMM and there are even cases where non-unit stride GEMMs are required \cite[Table~II]{Shi2016}.
Moreover, the index order of a temporary result influences possible mappings
to GEMM in nodes that involve the temporary result.

\subsubsection{Optimal index permutations}
Our approach to select the index orders of temporary tensors is to solve the
following discrete optimisation problem:
Let $\vertices(r)$ be the set of vertices of an AST with root $r$.
We denote the set of children of a vertex $v$ with $\children(v)$, and
the set of descendants of a vertex $v$ with $\descendants(v) = \vertices(v)\setminus\{v\}$.
For each vertex $v$ we have a set of permissible index permutations
$\mathcal{P}_{v}$.
We introduce variables~$x_v$, for all $v\in\vertices(r)$, where $x_v\in\mathcal{P}_v$
i.e.\ variable $x_v$ contains the index permutation for vertex $v$.
The Cartesian product of all permissible index permutations is called the 
configuration space.
The optimisation problem is to find a configuration which minimises a cost function
over all possible configurations, that is,
\begin{equation}\label{eq:permopt}
  c^* = \min_{x_r\in\mathcal{P}_r,x\in S}\cost\left(x_r, x\right),
  \;\; \text{ where } S = \bigtimes_{v\in\descendants(r)}\mathcal{P}_v.
\end{equation}
We define the cost function of a tree with root $r$ recursively in the following way:
\begin{equation}\label{eq:permcost}
\cost\left(x_r, x\right) =
  \nodecost_r\left(x_r,(x_c)_{c\in\children(r)}\right) +
  \sum_{c\in\children(r)} \cost(x_c,(x_d)_{d\in\descendants(c)}).
\end{equation}
The functions $\nodecost_r$ can be thought of as a measure
of the cost of the operation represented by a vertex in the AST.
We give a detailed definition of $\nodecost_r$ in \Cref{sec:costfun}.

We split the minimisation in two stages.
In the first stage we minimise over the root variable and its children,
in the second stage we minimise over all other variables, i.e.\ the variables
associated with the vertices in the set $\granddescendants(r) = \bigcup_{c\in\children(r)}\descendants(c)$. 
One can then show that
\begin{multline*}
c^{*} = \min_{x_r,(x_c)_{c\in\children(r)}}
	 \min_{(x_d)_{d\in\granddescendants(r)}}
	 \cost\left(x_r,(x_d)_{d\in\descendants(r)}\right) \\
       = \min_{x_r,(x_c)_{c\in\children(r)}} \left(
	    \nodecost_r\left(x_r,(x_c)_{c\in\children(r)}\right) +
	    \sum_{c\in\children(r)} f_{c}(x_c)\right),
\end{multline*}
where we introduced
\begin{equation}\label{eq:subproblem}
f_{c}(y) = \min_{(x_d)_{d\in\descendants(c)}}\cost\left(y,(x_d)_{d\in\descendants(c)}\right).
\end{equation}
The sub-problem introduced in \eqref{eq:subproblem} can be interpreted as finding
the optimal configuration for a sub-tree, assuming that the index permutation of
the root node is fixed.
In fact, $c^* = \min_{x_r\in\mathcal{P}_r}f_r(x_r)$.
Such problems are said to have an optimal substructure, because the optimal
solution can be constructed out of optimal solutions to sub-problems \cite{Cormen2009}.

Dynamic programming is a commonly used approach for problems with optimal substructure.
We consider a bottom-up dynamic programming algorithm, which works in the following way:
The AST is traversed in post-order.
For each vertex $v$, we solve problem \eqref{eq:subproblem} for every permissible
index permutation $x_v\in\mathcal{P}_v$.
The minimum cost as well as a minimising configuration is
memoized
in a dictionary.
If vertex $v$ is a leaf node, we simply evaluate all index permutations in
$\mathcal{P}_v$.
For internal nodes, we evaluate all configurations in
$\mathcal{P}_v\times(\bigtimes_{w\in \children(v)}\mathcal{P}_w)$.
In order to evaluate the cost, we evaluate the function $w_v$ and
for the sub-problems we look-up the memoized costs.
(The latter are available due to the post-order traversal.)
The run-time of this algorithm can be bound with $\mathcal{O}(N(n!)^{1+c})$,
where $N$ is the number of vertices in an AST, $n$ is the maximum number of indices in a vertex
(and hence $n!$ is the maximum size of each permissible index permutation set),
and $c$ is the maximum number of children.

\subsubsection{Cost function}\label{sec:costfun}
The missing pieces in the last section are the cost functions $\nodecost_r$.
Here, we choose a simple heuristic.
Our basic assumption is that our ASTs will consist mostly of LoGs and that
those will dominate the run-time.
Hence, all other operations are assigned cost zero.
Further assumptions are listed in the following:
\begin{itemize}
 \item Non-unit stride GEMMs are inferior to unit stride GEMMs.
 \item Transposes of $A$ ($B$) in the GEMM $AB$ should be avoided when
       using column-major (row-major) layout.
       Transposes of $B$ ($A$) should be avoided due to missing
       support in code generation back-ends.
 \item Large GEMMs are faster than small GEMMs, i.e.\ one should fuse as many indices as possible.
\end{itemize}
From these assumptions we define the cost of a LoG as the 4-tuple
$(s,l,r,-f)$, where $s$ is the number of required slices with non-unit stride,
$l$ is the number of left-transposes (assuming column-major layout),
$r$ is the number of right-transposes, and $f$ is the number of fused indices.
Cost comparison is based on lexicographic comparison of
$(s,l+r,-f)$, where the lower number of left-transposes
is deciding when two 3-tuples are equal.

The cost function is then
$$
w_r(x_r, (x_c)_{c\in\children(r)}) = \left\{\begin{array}{cl}
					    \text{MinLoG}(x_r, (x_c)_{c\in\children(r)}) & \text{ if type}(x_r) = \text{LoG}, \\
					    (\infty,\infty,\infty,\infty) & \text{ if type}(x_r) \neq \text{LoG } \wedge \text{violated constraints}, \\
                                            (0,0,0,0) & \text{ else,}
                                           \end{array}\right.
$$
where MinLoG enumerates feasible LoG-implementations and returns a LoG-implementation with minimum cost
or infinite cost if no mapping exists.
With ``violated constraints`` we summarise additional constraints
that exist, e.g.\ for Add or Assign nodes.

\subsubsection{Discussion}
The cost function we employ is clearly limited, as it is a heuristic and makes sense
for LoG nodes only.
However, choosing another cost function does not change the dynamic programming scheme,
as long as the cost function is structurally equivalent to \Cref{eq:permcost}.
For example, one could swap our LoG-cost by a run-time model based on micro-benchmarks,
using the methods developed in \cite{Peise2012}.

The algorithm needs to visit every node only once, hence the algorithm is feasible also
for large ASTs.
Problematic are tensors with many dimensions (say large $n$),
as the number of permutations is $n!$ and we need to check all of them.
But as stated in the introduction, we assume that our tensors fit into low-level
caches, which constrains $n$.
E.g.\ $n=6$ is the maximum number of dimensions we may choose such that a $4^n$
tensor (double precision) still fits into a 32\,KiB L1 cache.
In other words, \yateto~is not designed for high-dimensional tensors.

An example of an AST after mapping to LoG can be seen in \Cref{fig:final_ast}.

\subsection{Prefetching}
The library LIBXSMM, a generator for small matrix matrix multiplications,
allows software prefetching instructions to be included in
the generated assembler code \cite{Heinecke2016b}.
This may improve performance, particularly on Intel Knights Landing architecture \cite{Heinecke2016a,Uphoff2017}.
A possible prefetching strategy is to insert \texttt{vprefetch1}
instructions for a matrix B after writes to the result matrix C,
where the memory offset calculated for the matrix C are also used
for the matrix B.
The rationale is that B has the same or a similar shape as C.

In \yateto~the users may add a list of tensors they wish to be
prefetched to a kernel.
The FindPrefetchCapabilities visitor determines the number of bytes that may
be prefetched for every node of the kernel.
For a LoG-node this is equal to the number of bytes in the result tensor.
(Other nodes have no prefetch capability, but this may be added if
appropriate code generators are available.)
AssignPrefetch then greedily assigns each to-be-prefetched tensor P
to one of the nodes, such that the number of bytes of P matches the 
prefetch capability of the node.

\subsection{Control flow graph}
Following \Cref{fig:overview}, we convert the AST to a control flow graph (CFG) representation \cite{Seidl2012}.
The aim here is first to obtain a standardised sequential representation for the later
code generation step, and second to manage temporary buffers or to avoid them at all.

The data structure of the CFG are program points (vertices), which save the current state
of the kernel, and actions (edges), which transform the kernel's state.
An action may be either an assignment or a plus-equals operation.
The left-hand side is always a tensor.
The right-hand side may be either a tensor or one of the operations
modelled in the AST (such as LoG or Product).
The right-hand side may be augmented by multiplication with a scalar.

The individual optimisation steps use standard techniques from compiler optimisation \cite{Seidl2012},
thus we only briefly summarise them in the following:
\begin{enumerate}
 \item MergeScalarMultiplications: The actions $A = f(\dots); B = \alpha A$ are replaced by $B = \alpha f(\dots)$.
 \item LivenessAnalysis: Determine live variables \cite{Seidl2012}.
 \item SubstituteForward: If possible, replace $A$ by \texttt{tmp} after action $A = \mathtt{tmp}$.
 \item SubstituteBackward: If possible, replace left-hand side \texttt{tmp} by $A$ if there follows an action $A = \mathtt{tmp}$.
 \item RemoveEmptyStatements: Remove statements such as $A=A$.
 \item MergeActions: Try to merge actions, e.g.\ merge $A = f(\dots)$ and $B\mathrel{+}=A$ to $B\mathrel{+}=f(\dots)$,
	when there is no intermediate action which depends on $A$.
 \item DetermineLocalInitialization: Greedy map of temporary buffers to temporary variables.
       If buffers are used for multiple variables, set the buffer size to the maximum size
       required by its assigned variables.
\end{enumerate}

In order to illustrate the individual steps, we present matrix multiplication as an example:
\begin{codebox}
  \input{code/cfg_gemm.pyg}
\end{codebox}
\begin{figure}
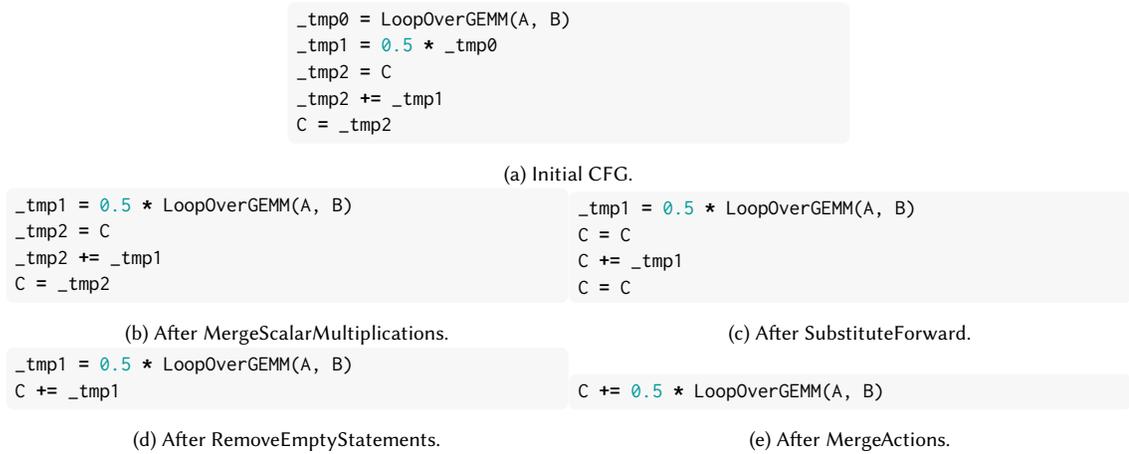

  \begin{subfigure}[b]{.49\linewidth}
    \centering
\begin{codebox}
  \input{code/cfg_gemm_1.pyg}
\end{codebox}
    \caption{Initial CFG.}\label{fig:cfgex_ini}
  \end{subfigure} \\
  \begin{subfigure}[b]{.49\linewidth}
    \centering
\begin{codebox}
  \input{code/cfg_gemm_2.pyg}
\end{codebox}
    \caption{After MergeScalarMultiplications.}
  \end{subfigure}
  \begin{subfigure}[b]{.49\linewidth}
    \centering
\begin{codebox}
  \input{code/cfg_gemm_3.pyg}
\end{codebox}
    \caption{After SubstituteForward.}
  \end{subfigure}
  \begin{subfigure}[b]{.49\linewidth}
    \centering
\begin{codebox}
  \input{code/cfg_gemm_4.pyg}
\end{codebox}
    \caption{After RemoveEmptyStatements.}
  \end{subfigure}
  \begin{subfigure}[b]{.49\linewidth}
    \centering
\begin{codebox}
\input{code/cfg_gemm_5.pyg}
\end{codebox}
    \caption{After MergeActions.}\label{fig:cfgex_final}
  \end{subfigure}
  \caption{Illustration of CFG optimisation of a matrix multiplication example.}
\end{figure}
A simple, but correct, traversal of the corresponding AST yields the sequential
CFG shown in \Cref{fig:cfgex_ini}.
This inefficient CFG is transformed to a single call to GEMM (with $\alpha=0.5$ and $\beta=1$),
as shown in \Cref{fig:cfgex_final}.

\subsection{Code generation}
The final step is to generate machine code.
In principle, the aim is to re-use existing code generators, such as LIBXSMM \cite{Heinecke2016b}
or PSpaMM \cite{Brei2018,Wauligmann2019}, or flavours of BLAS.
Besides, \yateto~generates C++-glue-code, and it is able to generate generic fallback code,
whenever an efficient implementation is not available.
The fallback code consists of nested loops and is architecture independent,
but we expect the performance to be rather poor and compiler-dependent.

Operations are divided into four types: copyscaleadd, indexsum, product, and log.
The first type is a combination of the BLAS routines COPY and AXPY, and may be implemented
in terms of these two functions.
The product and indexsum types correspond to the multiplication formulae and summation formulae
presented in \Cref{sec:hll}.
Having a generic implementation for these two types ensures that we cover a large class
of tensor operations \cite{Lam1997}.
The log type implements Loop-over-GEMMs.

Our concept requires that it should be simple to include additional code generators
depending on the type of operation and on the architecture.
We use the factory method pattern \cite{Gamma1995} to distinguish between code generators.
For each type, the factory method may return a different code generator depending on the
operation's description (e.g.\ dense or sparse matrices) and the architecture.
For the LoG type, \yateto~offers a generic C++-loop-based implementation, which internally
calls another factory method for GEMM.
But our structure would also allow to call specialised LoG generators, such as the ones
developed in \cite{Breuer2017} for contractions involving 3D tensors and matrices.

\subsection{Application interface}
A class is generated for every kernel.
The class has an argument-free execute function and pointers to the input and output tensors of
the kernel are public members, where the name of a member variable is given by the name of the tensor.
So in order to invoke a kernel one has to create a kernel object, set the pointers
to the input and output tensors, and call the execute function.
A kernel class also stores the following information as static members:
The minimal number of flops,\footnote{Strictly speaking, the term ``minimal''
is incorrect due to Strassen's algorithm or other bounds on the number 
of operations. That is, minimal is understood as minimal w.r.t.\ the cost
function defined in \Cref{sec:strengthred}.} also sometimes called non-zero flops,
and the number of flops the implementation requires, also sometimes
called hardware flops.
The number of non-zero flops is computed by means of strength
reduction with the sparse cost function in \Cref{sec:strengthred},
and the hardware flops are returned by the code generators.
These flop counters are useful to evaluate the performance of the kernels.

In addition to the kernel classes, a unit test is generated for every kernel.
The unit test compares the optimised kernel to the naive implementation
of the same kernel.
If the Frobenius norm of the difference of both implementations matches
up to a tolerance, the unit test passes.
(We allow a tolerance due to finite arithmetic imprecision caused
by the difference in order of addition.)

Moreover, a class is generated for every tensor,
where information about the memory layout is stored.
The most useful static member, which is declared
\texttt{constexpr}, is the number of required floating point numbers.
This information may be used to allocate stack or heap memory for a tensor.
For convenience, each tensor class contains a function that returns
a view object for a memory region.
Using a view object $V$ the user may access a d-dimensional tensor
with \texttt{operator()}, i.e.\ \texttt{V(i1,...,id)}.
Tensors whose entries are known at compile time
(e.g.\ stiffness matrices) may be stored within the generated code
as static array, already formatted according to the selected memory layout.

To integrate \yateto~into the build process of an application,
users need to have Python~3 and NumPy installed.
The generated code follows C++11 syntax and depends on a small header-only library,
i.e.\ the generated code itself requires only a recent C++-compiler.
The kernels may also be called from C or Fortran, by manually wrapping the kernel
calls in a C function.
We plan to generate a C and Fortran interface automatically in the future.

\section{Application to an ADER-dG method}\label{sec:applications}
We integrated \yateto~in two codes in order to evaluate practicability of our design,
as well as to evaluate the performance and optimisation opportunities.
Both codes employ a discontinuous Galerkin method with ADER time-stepping
based on the discrete Cauchy-Kovalewski procedure~\cite{Kaeser2007}.
Furthermore, both codes solve a linear PDE in two or three dimensions of the form of \Cref{eq:linearpde}.
The differences lie in the PDE (elastic wave equation, acoustic wave equation) and
the underlying finite elements (tetrahedra, rectangles).

In the following, we do not derive the numerical schemes but only indicate changes
w.r.t.\ literature references, in order to stay within the scope of this paper.
We validated both codes with a plane wave problem (e.g.\ \cite{Kaeser2007}),
and checked that the code converges (for an appropriately chosen time-step)
and that the observed convergence order matches the theoretical order
of the scheme.
Moreover, we tried to keep the amount of indices to a minimum by mostly omitting
dependencies on space and time.
We use the Einstein convention for subscript indices, but not for superscript indices.

\subsection{SeisSol}\label{sec:appseissol}
The earthquake simulation code SeisSol (\url{www.github.com/SeisSol/SeisSol}) solves the seismic wave equation
with support for elastic \cite{Dumbser2006}, viscoelastic \cite{Kaeser2007},
and viscoplastic rheological models \cite{Wollherr2018}.
The code operates on unstructured tetrahedral meshes and the quantities are
approximated with Dubiner's basis functions.
In essence, one may write the computational kernels as matrix chain products
of small matrices, where small means that the maximum dimension size of each
matrix is smaller or equal $\mathcal{B}={N+3 \choose 3}$, where $N$ is the maximum
polynomial degree of the basis functions and $\mathcal{O}=N+1$ is the theoretical
convergence order.
SeisSol already makes use of the code generator tailored for matrix chain products,
which simplifies the implementation of the various rheological models \cite{Uphoff2016}.

For elastic rheological models, the code achieves about 50\,\% of peak performance on compute
clusters based on Intel's Haswell or Knight's Landing architecture \cite{Heinecke2016a,Uphoff2017}.
Moreover, the authors of \cite{Breuer2017} investigate the possibility to
fuse multiple simulations in a single run, using the same ADER-DG method as SeisSol
and an elastic rheological model.
They report a speed-up of $2.1\times$ for a fourth-order scheme, when fusing
8 simulations in one run (compared to 8 individual simulations).
Their main optimisation ingredient is a specialised code generator for tensor contractions involving
3D tensors and sparse matrices.

\subsubsection{ADER-DG}
The numerical scheme of SeisSol is given in detail in \cite{Dumbser2006,Kaeser2007,Uphoff2017},
and is stated in a compact form in \Cref{eq:seissol:derivative,eq:seissol:timeint,eq:seissol:correction_local,eq:seissol:correction_neigh}
for an elastic rheological model.
The equations already include multiple simulations, as in~\cite{Breuer2017}.
Multiple simulations are easily modelled using the Einstein convention as one simply has to add
another index to the degrees of freedom, which we have denoted with a blue s.

The first step in the numerical scheme is to compute the time-derivatives
of the degrees of freedom $Q$, using a Cauchy-Kovalewski procedure:
\begin{equation}
  \mathcal{D}_{\ms kp}^{(\delta+1)}
  := \tilde{K}_{kl}^\xi\mathcal{D}_{\ms lq}^{\delta} A^*_{pq}
        + \tilde{K}_{kl}^{\eta}\mathcal{D}_{\ms lq}^{\delta} B^*_{pq}
        + \tilde{K}_{kl}^{\zeta}\mathcal{D}_{\ms lq}^{\delta} C^*_{pq},
        \text{ where }\,\, \mathcal{D}_{\ms kp}^{0}:=Q_{\ms kp}^{\mu\tau}. \label{eq:seissol:derivative}
\end{equation}
The indices $\mu$ and $\tau$ denote the space-time element
and index $\delta$ denotes the order of the time-derivative.
The matrices $\tilde{K}$ are sparse $\mathcal{B}\times \mathcal{B}$ matrices, and $A^*, B^*,$ and $C^*$ are sparse $9\times 9$ matrices,
which are element-dependent linear combinations of the Jacobians $A,B,$ and $C$ from \Cref{eq:linearpde}.
For every simulation one needs to store a $\mathcal{B}\times 9$ matrix for the degrees of freedom $Q$
and the derivatives $\mathcal{D}$, yielding $\mathcal{S}\times \mathcal{B}\times 9$ tensors, where $\mathcal{S}$ is the number of simulations.

The evolution in time is predicted with a Taylor expansion, without
considering an element's boundaries.
In the DG-scheme, the time-integral over a time-step with size $\Delta t$ is required,
which is computed as
\begin{equation}
\mathcal{I}_{\ms kp}
  :=\; \sum_{\delta=0}^N \dfrac{\Delta t^{\delta+1}}{(\delta+1)!} \mathcal{D}_{\ms kp}^{\delta}. \label{eq:seissol:timeint}
\end{equation}
The time-integral is inserted into the DG-scheme, yielding the update scheme for
the degrees of freedom.
\begin{align}
  Q_{\ms kp}^{\mu(\tau+1)}
    :=&\; Q_{\ms kp}^{\mu\tau} + \hat{K}_{kl}^\xi \mathcal{I}_{\ms lq} A^*_{pq}
     + \hat{K}_{kl}^{\eta} \mathcal{I}_{\ms lq} B^*_{pq}
     + \hat{K}_{kl}^{\zeta} \mathcal{I}_{\ms lq} B^*_{pq}
     + \sum_{\sigma=1}^4 \hat{R}_{km}^{\sigma}\tilde{R}_{lm}^{\sigma} \mathcal{I}_{\ms lq} (A^{+})^{\sigma}_{pq} + \label{eq:seissol:correction_local} \\
     &+ \sum_{\sigma=1}^4 \hat{R}_{km}^{\sigma}f_{mn}^{\text{Rot}(\sigma,\mu)}R_{ln}^{\text{Face}(\sigma,\mu)} \mathcal{I}_{\ms lq}^{\text{Neigh}(\sigma,\mu)} (A^-)^{\sigma}_{pq}.\label{eq:seissol:correction_neigh}
\end{align}
The matrices $\hat{K}$ are sparse $\mathcal{B}\times \mathcal{B}$ matrices,
$\hat{R}, \tilde{R},$ and $R$ are sparse $\mathcal{B}\times \tilde{\mathcal{B}}$ matrices,
where $\tilde{\mathcal{B}} = {N+2\choose 2}$, and $f$ are sparse $\tilde{\mathcal{B}}\times\tilde{\mathcal{B}}$
matrices.
$A^+$ and $A^-$ contain the solution to the 1D Riemann problem which is solved
along the four sides of a tetrahedron, and these are dense $9\times 9$ matrices.

\subsubsection{Equivalent sparsity patterns}
The matrices $\hat{K}$ contain large zero blocks,
especially only the first $\tilde{\mathcal{B}}$ columns are non-zero.
\yateto~automatically determines that the EQSPP of $\mathcal{I}$
is a $\mathcal{S}\times \tilde{\mathcal{B}}\times 9$ subtensor.
Hence, the contraction $\mathcal{I}_{sl[q]}\hat{K}_{kl}^T$ is
mapped to a GEMM of size $(M,N,K) = (\mathcal{S}, \mathcal{B}, \tilde{\mathcal{B}})$
instead of a GEMM of size $(\mathcal{S}, \mathcal{B}, \mathcal{B})$.
Also the contraction $\mathcal{I}_{(sl)q}(A^*)_{pq}^T$
is automatically mapped to a GEMM of size
$(\mathcal{S}\tilde{\mathcal{B}}, 9, 9)$ instead of a GEMM of size
$(\mathcal{S}\mathcal{B}, 9, 9)$.
So the percentage of original flops scales with $3/(N+3)$, e.g.\ 50\,\%
for a fourth-order scheme or 37.5\,\% for a sixth-order scheme.

The sparsity pattern of $\tilde{K}$ is equivalent to the sparsity
pattern of $\hat{K}^T$ and it has a staircase structure.
The effect is, that derivative $D^\delta$ has only
${N-\delta+3 \choose 3}$ non-zero coefficients per quantity \cite{Breuer2014}.
We may set the sparsity patterns of the derivatives $D^\delta$
accordingly in order to exploit the vanishing coefficient -- or
we use existing visitors of \yateto~in order to simultaneously 
define the derivative kernels and determine the sparsity patterns
of the derivatives automatically, as shown in \Cref{fig:automatic_derivative_spp}.

\begin{figure}
\begin{codebox}
 \input{code/derivatives.pyg}
\end{codebox}
\caption{Exemplary code which models vanishing coefficients in
         the Cauchy-Kovalewski procedure from \Cref{eq:seissol:derivative}.
         First a partial AST is built and stored in \texttt{derivativeSum}.
         Then the first two visitors in the transformation process shown
         in \Cref{fig:overview} are used to obtain the sparsity pattern
         of derivative $D^i$.
         Finally, a tensor using the derived sparsity pattern is defined
         and a complete AST is added to the code generator.
        }\label{fig:automatic_derivative_spp}
\end{figure}

\subsubsection{Strength reduction}
By solving the matrix chain multiplication order problem
the number of required flops in \Cref{eq:seissol:correction_neigh} can be reduced from $\mathcal{O}(N^6)$
to $\mathcal{O}(N^5)$ \cite{Uphoff2017}.
The optimal matrix chain order
$\hat{R}_{ka}\left(\left(f_{ab}\left(R_{lb} \mathcal{I}_{lq}\right)\right) A^-_{pq}\right)$
is reproduced by the implementation of the strength reduction algorithm.
Interestingly, for 8--32 fused simulations the order of evaluation
$\left(\hat{R}_{ka}f_{ab}\right)\left(\left(R_{lb} \mathcal{I}_{\ms lq}\right) A^-_{pq}\right)$
is optimal, found automatically by strength reduction.

\subsubsection{Optimal index permutations}
We take \Cref{eq:seissol:correction_neigh} as example,
but note that the same discussion can be applied to the other kernels.
From \yateto~we obtain the following mappings to LoG for either a single 
simulation or for multiple simulations (Greek letters denote temporary tensors):
\begin{equation*}
\begin{aligned}[c]
 &\text{Single} \\
 \alpha_{nq} &:= R_{ln}^T \mathcal{I}_{lq} \\
 \beta_{mq} &:= f_{mn} \alpha_{nq} \\
 \gamma_{mp} &:= \beta_{mq}(A^-)_{pq}^T \\
 Q_{kp} &:= Q_{kp} + \hat{R}_{km}\gamma_{mp}
\end{aligned}
\qquad\qquad\qquad\qquad
\begin{aligned}[c]
 &\text{Multiple} \\
 \alpha_{sn[q]} &:= \mathcal{I}_{sl[q]} R_{ln} \\
 \beta_{(sn)p} &:= \alpha_{(sn)q} (A^-)_{pq}^T \\
 \gamma_{kn} &:= \hat{R}_{km}f_{mn} \\
 Q_{sk[p]} &:= Q_{sk[p]} + \beta_{sn[p]} \gamma_{kn}^T
\end{aligned}
\end{equation*}
We observe that the optimal index permutation algorithm from \Cref{sec:optindperm} fulfils the
goals of the cost function:
Non-unit stride GEMMs are not present, transposes are kept to a minimum and indices
are fused if possible.

In SeisSol, the implementation is in fact a bit different:
The $A^-$ and $R$ matrices (as well as the star matrices and the $A^+$ matrices) are
stored transposed, such that a transpose-free scheme is obtained.
For multiple simulations we can also obtain a transpose-free scheme
in the following manner:
The $A^-$ matrices (as well the star matrices and the $A^+$ matrices)
are stored transposed as for a single simulation,
but for the other matrices we face the opposite situation:
The $R$ matrices are stored in normal order and we have to store
$\hat{R}$ and $f$ transposed as well as $\hat{K}$ and $\tilde{K}$ for the other kernels.

\subsubsection{Findings}
\yateto~is able to reproduce the major optimisations from SeisSol:
Zero-blocks are exploited and the optimal matrix chain multiplication order is found.
For multiple simulations, zero blocks are also exploited automatically,
strength reduction revealed the optimal---and different---evaluation order,
and a transpose-free scheme is found by inspection of mappings to GEMM
produced by \yateto.

\subsection{LinA}
The code LinA was developed for education purposes, and contains a basic
implementation of the ADER-DG method.
LinA solves the linearised equations of acoustics in
two dimensions \cite{Leveque2002}, which involve the three quantities
pressure and particle velocities in horizontal and vertical direction.
A uniform mesh with rectangular elements is used.

Our goal here is to evaluate the application of \yateto~to a
DG-spectral-element-like method (cf.~\cite[Chapter~5.4]{Kopriva2009}
for an introduction to DG-SEM).
As briefly introduced in \Cref{sec:intro_dg},
quantities are represented using a tensor product basis,
i.e.\ $q_p = Q_{lmnp} \psi_l(x) \psi_m(y) \psi_n(z)$.
An advantage of DG-SEM is that integrals over the
reference elements may be split in a tensor product of 1D integrals.
E.g.\ the computation of the gradient in the weak form
can be done with $\mathcal{O}(N^{d+1})$ instead of $\mathcal{O}(N^{2d})$
computations per element, where $d$ is the number of dimensions.
On the downside, one needs to store $(N+1)^d$ degrees of freedom
per quantity instead of~${N+d \choose d}$ per quantity,
which is the minimum amount of degrees of freedom needed
to represent the same space of polynomials.

We implemented the ADER-DG-SEM method in LinA using \yateto~(\url{www.github.com/TUM-I5/LinA}),
and extended the code to three dimensions, which implies
an additional quantity or 4 quantities in total.
For the 1D basis functions $\psi_i$ we use a nodal basis with Gauss-Lobatto
points and Legendre polynomials~\cite{Hesthaven2008}.
In the following we are going to abuse notation and use the
same symbols for tensors as in SeisSol, even though the
tensors are of different shape and have different sparsity patterns.

\subsubsection{Numerical scheme}
The first step in LinA, as in SeisSol, is to predict the element-local
evolution in time:
\begin{align}
\mathcal{D}_{xyzp}^{(\delta+1)}
  :=&\; \tilde{K}_{xl}\mathcal{D}_{lyzq}^{\delta} A^*_{pq}
        + \tilde{K}_{ym}\mathcal{D}_{xmzq}^{\delta} B^*_{pq}
        + \tilde{K}_{zn}\mathcal{D}_{xynq}^{\delta} C^*_{pq},
        \text{ where }\,\, \mathcal{D}_{xyzp}^{0}:=Q_{xyzp}^{\xi\upsilon\zeta\tau}, \label{eq:lina:derivative}\\
\mathcal{I}_{xyzp}
  :=&\; \sum_{\delta=0}^N \dfrac{\Delta t^{\delta+1}}{(\delta+1)!} \mathcal{D}_{xyzp}^{\delta}, \label{eq:lina:integration}
\end{align}
Indices $\xi,\upsilon,\zeta$ denote the location of an element
on the grid and $\tau$ is the time index.
The matrices $\tilde{K}$ are $(N+1)\times (N+1)$ matrices
and $A^*$ and $B^*$ are sparse $4\times 4$ matrices.
Both $Q$ and $\mathcal{D}$ are $(N+1)\times (N+1)\times (N+1)\times 4$ tensors.

The weak derivatives are evaluated afterwards:
\begin{equation}
Q^*_{xyzp}
    := Q_{xyzp}^{\xi\upsilon\zeta\tau} + \hat{K}_{xl} \mathcal{I}_{lyzq} A^*_{pq}
     + \hat{K}_{ym} \mathcal{I}_{xmzq} B^*_{pq}
     + \hat{K}_{zn} \mathcal{I}_{xynq} C^*_{pq}
\end{equation}

In contrast to SeisSol, the neighbour exchange is not based on volume data.
Instead, it is sufficient to consider the boundary nodes.
We copy the boundary nodes to continuous storage for each side, such that we need to
store eight 3D tensors instead of a single 4D tensor.
In \yateto, we represent the ``copy'' operation by the following
matrix-vector products:
\begin{equation}\label{eq:lin3d_evaluateSide}
 \mathcal{X}_{yzp}^{\sigma}
  := F_l^\sigma \mathcal{I}_{lyzp}, \quad
\mathcal{Y}_{xzp}^{\sigma}
  := F_m^\sigma \mathcal{I}_{xmzp}, \quad
\mathcal{Z}_{xyp}^{\sigma}
  := F_n^\sigma \mathcal{I}_{xynp},
\end{equation}
where $F_i^\sigma=\delta_{i0}$ for $\sigma\in\{\text{left},\text{bottom},\text{back}\}$,
and $F^1_i=\delta_{iN}$ for $\sigma\in\{\text{right},\text{top},\text{front}\}$.
The 2D boundary nodes are then used to apply the numerical fluxes:
\begin{multline}\label{eq:lina3d_flux}
Q_{xyzp}^{\xi\upsilon\zeta(\tau+1)} := Q^*_{xyzp} +
     \sum_{\sigma \in \{\text{left},\text{right}\}} \hat{F}_{x}^{\sigma}\left(\mathcal{X}_{yzq}^{\sigma} (A^+)_{pq}^{\sigma}
        + \mathcal{X}_{yzq}^{\text{Neigh}(\sigma,\xi,\upsilon,\zeta)} (A^-)_{pq}^{\sigma}\right) + \\
     + \sum_{\sigma \in \{\text{bottom},\text{top}\}} \hat{F}_{y}^{\sigma}\left(\mathcal{Y}_{xzq}^{\sigma} (A^+)_{pq}^{\sigma}
        + \mathcal{Y}_{xzq}^{\text{Neigh}(\sigma,\xi,\upsilon,\zeta)} (A^-)_{pq}^{\sigma}\right) + \\
     + \sum_{\sigma \in \{\text{back},\text{front}\}} \hat{F}_{z}^{\sigma}\left(\mathcal{Z}_{xyq}^{\sigma} (A^+)_{pq}^{\sigma}
        + \mathcal{Z}_{xyq}^{\text{Neigh}(\sigma,\xi,\upsilon,\zeta)} (A^-)_{pq}^{\sigma}\right)
\end{multline}

\subsubsection{Implementation}
Mapping the numerical scheme to our DSL is simple as one may almost copy
the above formulation in Einstein notation.
E.g.\ the last line in \Cref{eq:lina3d_flux} can be written as following:
\begin{codebox}
 \input{code/lina.pyg}
\end{codebox}
Furthermore, by inspecting the generated implementation we decided
to store the matrices $A^*,B^*,C^*,A^+,A^-$ transposed
and also a transposed copy of $\hat{K}$ and $\tilde{K}$,
which leads to a transpose-free scheme.
The eventual C++-code for the last line in \Cref{eq:lina3d_flux}
can be seen in \Cref{fig:lina_impl}.
There, a generic nested loop-code
implements the outer product $Q_{xyzp} := \text{tmp}_{xyp}\hat{F}_{z} + Q_{xyzp}$.
In principle, one might also map this operation to BLAS, e.g.\ in terms
of the level 2 routine GER, which implements a rank-1 update of a matrix.
That is, one could implement the outer product as a Loop-over-GER,
i.e.\ $Q_{(xy)z[p]} := \text{tmp}_{(xy)[p]}\hat{F}_{z} + Q_{(xy)z[p]}$.
In \yateto~such an implementation corresponds to adding another back-end.
However, we leave the implementation and performance evaluation of Loop-over-GER
for future work.

\begin{figure}
\begin{codebox}
 \input{code/lina_impl.pyg}
\end{codebox}
 \caption{Sample C++-code generated in LinA.
 Indices in tensor contractions are fused such that only a call to GEMM is required.
 The GEMMs (prefix \texttt{pspamm}) are separate functions which contain inline
 assembler code, and are generated specifically for the kernel.
 The nested loops implement an outer product, and were 
 generated using the generic fallback code generator.}\label{fig:lina_impl}
\end{figure}

\section{Implementation aspects and performance results}
In this section we thoroughly evaluate the performance of the
applications presented in \Cref{sec:applications}.
From the results in this section we ultimately want to answer the
following questions:
First, is our tensor toolbox able to reproduce SeisSol's performance
for the special case of matrix chain products?
Second, can we achieve similar performance improvements as in \cite{Breuer2017}
by fusing multiple simulations, but using our generic approach?
Third, do we still get high performance when changing the PDE,
the finite element type, and the basis functions?

\subsection{Hardware and software environment}
We run all our performance tests on two recent Intel architectures:
\begin{itemize}
 \item[\textbf{KNL}] Intel Xeon Phi 7250 on Stampede~2, ``Knights Landing``,
 single socket configuration with 68 cores,
 1.4\,GHz nominal frequency, 1.6\,GHz single-core turbo frequency,
 theoretical peak performance of 3.0\,TFLOPS using double precision
 or 6.1\,TFLOPS using single precision (at 1.4\,GHz),
 1\,MiB private L2 cache per tile (shared between two cores).
 \item[\textbf{SKX}]  Intel Xeon Platinum 8174 on SuperMUC-NG, ``Skylake``,
 dual socket configuration with 24 cores per socket,
 3.1\,GHz nominal frequency,
 2.7\,GHz AVX-512 all-core turbo frequency,\footnote{We could not find
 a processor specification as the 8174 is a special model. The system's peak
 performance is claimed to be 26.3\ PFLOPS on 6336 nodes
 (\url{https://doku.lrz.de/display/PUBLIC/Hardware+of+SuperMUC-NG}),
 which would imply a peak AVX-512 frequency of 2.7\,GHz.
 A test measuring the time of around 54\,billion
 \texttt{vfmadd231pd} instructions (using alternating \texttt{zmm} registers)
 indicates that a peak frequency of 2.79\,GHz is possible.
 In this paper we adopt the vendor's claim of 2.7\,GHz, but 
 we encourage the reader to scale the results with a
 peak frequency that may be more appropriate.}
 theoretical peak performance of 4.1\,TFLOPS using double precision
 or 8.3\,TFLOPS using single precision,
 1\,MiB private L2 cache per core,
 1.375\,MiB shared L3 cache per core (non-inclusive).
\end{itemize}

In each experiment, threads are pinned to cores using the \texttt{KMP\_AFFINITY} or \texttt{OMP\_PLACES}
environment variable. On SKX we use 48 threads and pin each thread to exactly one core.
For simultaneous multithreading (SMT) we use a compact pinning with 2 threads per core.
On KNL we basically employ the same strategy with a maximum of 2 threads per core,
but we leave the first tile free for the operating system.
Thus we use 66 threads without SMT and 132 threads with SMT.
Furthermore, on KNL we use the cache memory mode for SeisSol
and the flat memory mode for LinA, where memory is bound to
MCDRAM using \texttt{numactl}.

Our software is compiled with the Intel C++-compiler (version~18).
We generate our GEMM kernels with LIBXSMM (version~1.10, \cite{Heinecke2016b})
and PSpaMM (commit f59f98d, \cite{Wauligmann2019}),
which both offer highly tuned GEMMs for KNL and SKX.
LIBXSMM is used in static mode (instead of the just-in-time-compilation mode),
i.e.\ code is generated at compile-time.
We remark that the generated file \texttt{subroutine.cpp}, which
contains kernels from LIBXSMM and PSpaMM,
is compiled with option \texttt{-mno-red-zone}.
Otherwise, the compiled code might be wrong due to a combination of
(1) the compiler does not inspect inline assembly,
(2) the generated code might modify the stack (\texttt{pushq}, \texttt{popq}),
(3) Linux calling conventions.

If not stated otherwise, we report the minimum run-time over at least 5 runs.

\subsection{Sparse matrices}
PSpaMM originated from a student research project for dense $\times$ sparse
matrix multiplications \cite{Brei2018}.
The starting point of the project is a clone of the almost optimal
register-blocked LIBXSMM kernels for KNL.
In principle, sparse matrix multiplication is then implemented
by removing every instruction that either is a multiplication with zero
or becomes unnecessary (such as loads of never used data).
Control structures such as the row indices array and the column pointer
array are completely unrolled.

An issue with unrolling control structures is that code is also data,
which needs to be fetched from memory or cache.
An experiment with random sparse matrices in \cite{Brei2018},
using a predecessor of PSpaMM, suggests
that the generated dense $\times$ sparse matrix multiplication
is always faster than its dense $\times$ dense counterpart---as long
as the number of non-zeros stays below a threshold such that code
stays in the L1i cache.
The threshold can be estimated by dividing the L1i cache size
by the size of a fused multiply-add instruction,
as a fused multiply-add instruction is generated for every non-zero
in the sparse matrix.

The sparse $\times$ dense case is more involved \cite{Breuer2013},
the fundamental issue being that the standard outer product formulation for
matrix multiplications forces to vectorise over sparse vectors.
In SeisSol, sparse $\times$ dense kernels did not improve time-to-solution
on KNL \cite{Heinecke2016a}.

As a consequence, we employ the following rule:
We use a sparse memory layout when a sparse matrix is multiplied from the
right and we use a dense memory layout when a sparse matrix is multiplied
from the left.

\subsection{SeisSol: Performance}\label{sec:seissol_reproducer}
\begin{figure}
  \hspace*{-.5cm}
  \includegraphics{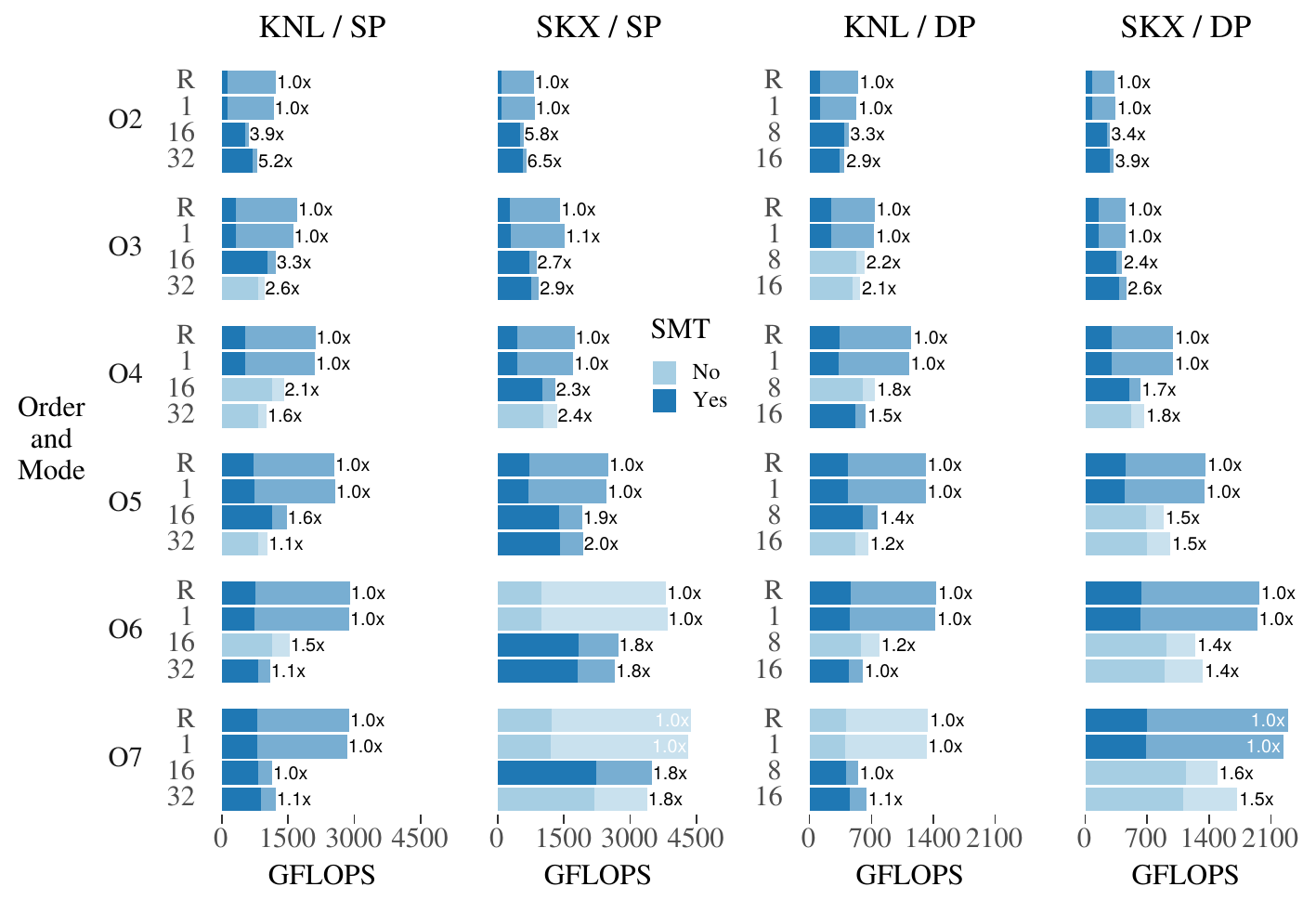}
  \caption{The figure shows performance experiments with SeisSol's kernels.
  We compare architectures (KNL, SKX), precision (SP=single, DP=double), orders (O2--O7),
  and the number of simultaneous simulations (1--32).
  The letter R (reference) denotes the latest SeisSol release \cite{Uphoff2017}.
  Dark bars show non-zero flops and light bars denote hardware flops
  and the shade of blue indicates if SMT increases performance.
  Numbers right to bars show speed-ups, if one compares S simultaneous
  simulations to S individual simulations.
  The benchmark results are tabulated in \Cref{apx:seissol}.
  }\label{fig:perf:seissol}
\end{figure}
In this section, we present and discuss the performance results for SeisSol
shown in \Cref{fig:perf:seissol}.
We note that SeisSol's performance reproducer is used~\cite{Uphoff2017}.
The latter calls the same kernels as SeisSol does, but
the reproducer operates on random initial data and the
cell neighbouring relations are random.

We investigate single precision arithmetic as well as double precision arithmetic.
(As \yateto~supports single and double precision, we take an agnostic viewpoint
on necessary precision and simply generate and test both precisions.)
Furthermore, we ran the performance reproducer of the latest SeisSol
release for reference \cite{Uphoff2017}.
Performance is measured in terms of non-zero flops and hardware flops.
The important measure here is non-zero flops, as it is independent
of the choice of sparse and dense memory layouts and as such
represents time-to-solution.
However, the non-zero flops measure assumes perfect usage of sparsity
patterns and we see it as too optimistic w.r.t.\ hardware.
Hence we include hardware flops as measure of exploitation of the hardware's
capabilities.

In \Cref{fig:perf:seissol} we observe that the results for single
simulations (row ``1'') closely match the performance of the reference (row ``R'')
over all orders, architectures, and precisions.
We conclude that the special case of matrix chain multiplications
is well handled within \yateto.
We also see that there is a large gap between non-zero flops
and hardware flops.
The latter is due to handling sparse matrices as dense, which minimises
time-to-solution according to an auto-tuning procedure \cite{Heinecke2016a}.
For multiple simulations (rows ``8'', ``16'', and ``32''),
most sparse matrices are multiplied from the
right such that the favourable dense $\times$ sparse case is obtained.
Moreover, the arithmetic intensity is higher as matrices already
in cache are reused~\cite{Breuer2017}.
In our experiments, we observe that multiple simulations (compared to the same
amount of individual simulations) yield a speed-up of
1.1$\times$--3.9$\times$ for DP and 1.1$\times$--6.5$\times$ for SP.
Furthermore, we observe an increase in non-zero peak performance from
1.1\,\%--16.9\,\% to 6.7\,\%--27.5\,\%.
Specifically, a non-zero peak efficiency of 19.6\,\% for KNL O4 DP
is obtained, which is within a few percent of the non-zero peak efficiency
obtained in \cite{Breuer2017}.\footnote{We note that the schemes of
\cite{Uphoff2017} and \cite{Breuer2017} are not identical
due to a different handling of the boundary integral terms,
however the remainder of the scheme is identical.}

\subsection{SeisSol: Validation}
\begin{figure}
  \includegraphics{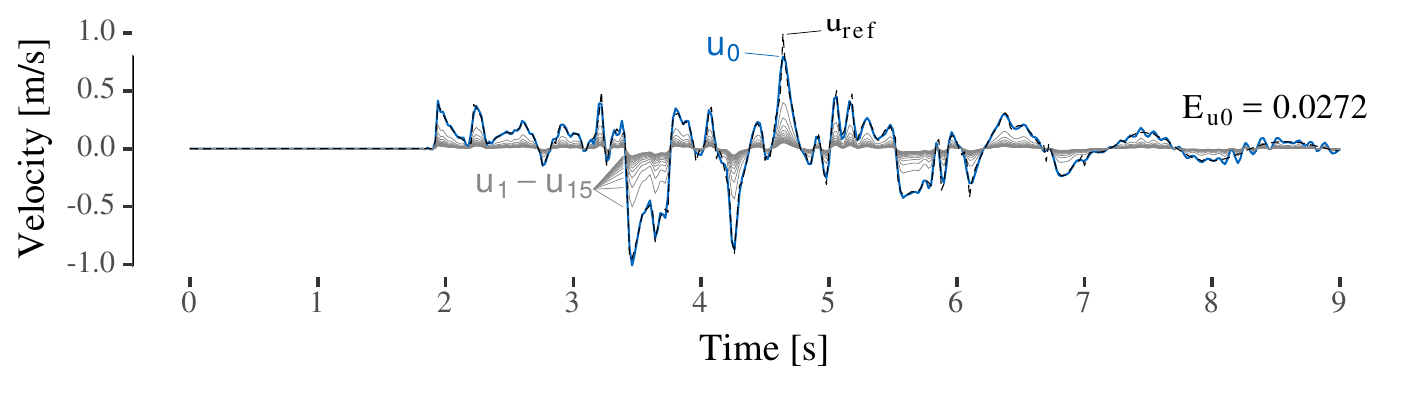}
  \caption{Comparison of simulation and reference solution.
  Shown is the particle velocity in x-direction of the 9-th receiver of the LOH.1 benchmark
  (see \url{http://www.sismowine.org/model/WP2_LOH1.pdf}).
  A total of 16 simulations ran simultaneously in single precision,
  where the 0-th simulation has the same seismic moment as the reference.
  Other simulations ($u_1$--$u_{15}$) have a reduced moment (see text).
  The relative seismogram misfit for the 0-th simulation is shown in $E_{u0}$.
  }\label{fig:valid:loh1}
\end{figure}

\begin{table}
\centering
\caption{Summary of performance results for the LOH.1 benchmark using 8 SKX nodes and O6.}\label{tab:loh1}
\begin{tabular}{llrrrrr}
  \toprule
Precision & LTS & Mode & NZ-GFLOPS & HW-GFLOPS & \% $\text{HW}_\text{proxy}$ & Time / Sim. [min] \\ 
  \midrule
DP & No & 1 & 524 & 1618 & 83 & 42.1 \\ 
  DP & No & 8 & 839 & 1148 & 92 & 26.6 \\ 
  DP & Yes & 1 & 480 & 1485 & 76 & 17.8 \\ 
  DP & Yes & 8 & 752 & 1030 & 83 & 11.5 \\ 
  SP & No & 1 & 840 & 3217 & 84 & 26.2 \\ 
  SP & No & 16 & 1719 & 2559 & 94 & 13.0 \\ 
  SP & Yes & 1 & 789 & 3023 & 79 & 10.8 \\ 
  SP & Yes & 16 & 1433 & 2134 & 78 & 6.0 \\ 
   \bottomrule
\end{tabular}
\end{table}

In \Cref{sec:seissol_reproducer} a performance reproducer is used to evaluate
SeisSol's performance.
Here, we evaluate the correctness of the implementation using the 
well-established Layer Over Halfspace 1 (LOH.1) benchmark \cite{Day2003}.
The model description as well as the reference solution is taken
from the SeISmic MOdeling Web INterfacE (\url{www.sismowine.org}).

We ran a total of 8 or 16 simulations simultaneously on a mesh
with 1.1\,million elements and Order 6.
The 0-th simulation has a seismic moment of $M_0=10^{18}\,\text{Nm}$,
as required by the benchmark.
Other simulations have a seimic moment of $M_0/(1+s)$ where $s$ is the simulation
number, i.e.\ $M_0/2$, $M_0/3$, $M_0/4$, etc.\
Due to the PDEs being linear, we expect that the amplitude of particle
velocities is reduced by the same factor.
\Cref{fig:valid:loh1} shows that the 0-th simulation matches the reference well,
with a low relative seismogram misfit.
The other simulations have their amplitude reduced, as expected.
We note that the relative seismogram misfit is identical up to 3 digits,
independent of the number of simulations or floating point precision.

The measured performance on 8 SKX nodes is shown in \Cref{tab:loh1}, where we 
evaluated single simulations and multiple simulations, SP and DP,
as well as Global Time Stepping (GTS) and Local Time Stepping (LTS).
SMT is enabled or disabled according to \Cref{fig:perf:seissol},
but only 47 cores for computation are used as one core is left
for a communication thread.
We observe that we obtain a lower performance in comparison to the
performance reproducer (see column \% $\text{HW}_\text{proxy}$).
We do not attempt a detailed comparison of the performance reproducer
with the real application, however, it is interesting to note that
for multiple simulations we ``lose'' less performance than for a single
simulation (e.g.\ 84\,\% vs.\ 94\,\% for GTS-SP).
The latter is a hint that the remaining serial part of SeisSol might be responsible
for the lower performance, as the serial part is mostly independent of the number
of simulations but the parallel part becomes more expensive for more simulations.
LTS is less efficient due to smaller loop lengths, but pays off in
time-to-solution.
In total, the expected speed-up of multiple simulations is reproduced
(1.6$\times$ for GTS-DP, 2.0$\times$ for GTS-SP).

\subsection{LinA}
\begin{figure}
 \includegraphics[width=\textwidth]{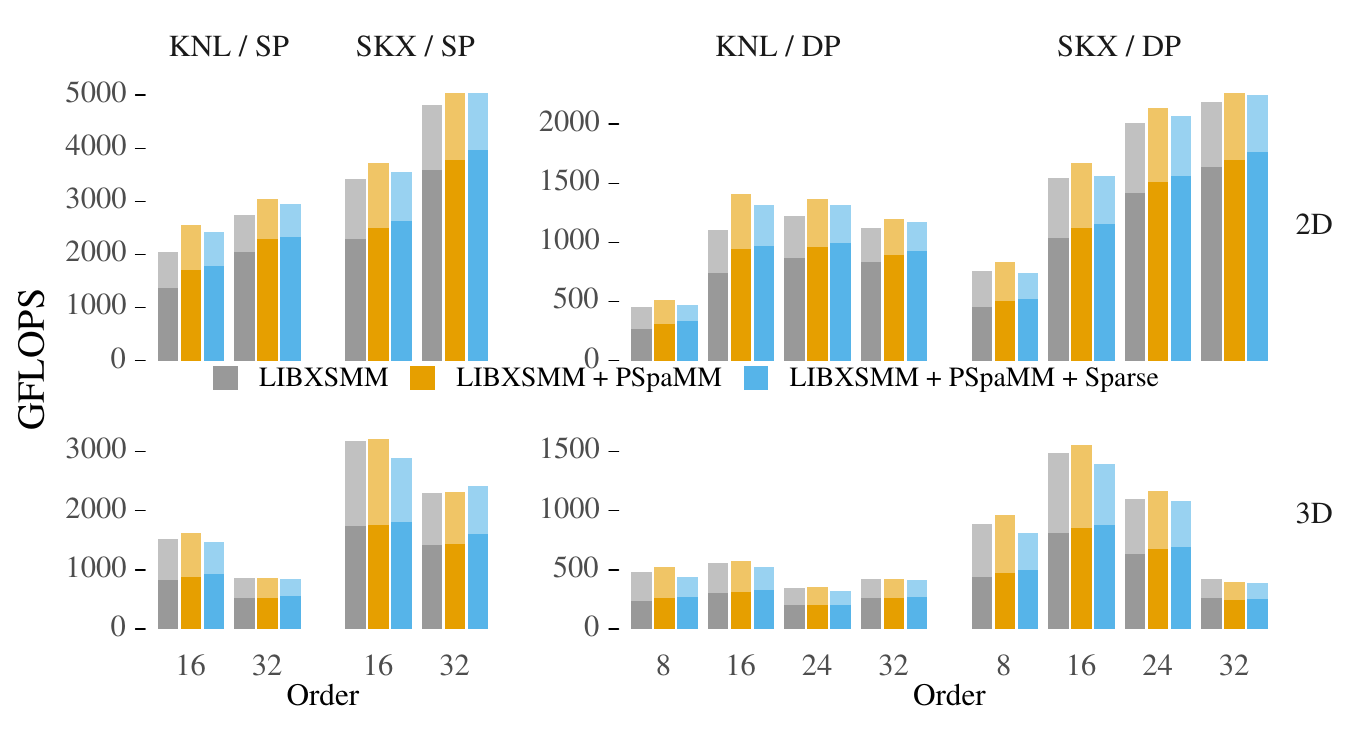}
 \caption{We show performance results of a DG-SEM method implemented in the
 code LinA. We compare architectures (KNL, SKX), precision (SP=single, DP=double),
 orders (8, 16, 24, 32), and number of spatial dimensions (2D, 3D).
 Dark bars show non-zero flops and light bars denote hardware flops.
 The colour indicates the employed code generators as well
 as the usage of the CSC memory layout for some matrices.}\label{fig:perf:lina}
\end{figure}
While LinA uses the same ADER-DG scheme as SeisSol, the implementation
is very different due to the use of tensor basis functions
and cuboid elements.
The main difference is that the arithmetic intensity only grows
linear with polynomial degree, instead of cubic growth as in SeisSol.
But due to the use of a nodal basis, most of the matrices in the
scheme are dense or almost dense and sparsity mainly comes in
due to the coefficient matrices of the PDE.
Another issue is the increased storage required by the degrees of
freedom, such that the last level cache size is becoming
an issue for three spatial dimensions, hence we also tested
the scheme in two spatial dimensions to better understand the tradeoff
between arithmetic intensity and cache usage.

A general observation in the performance results shown in
\Cref{fig:perf:lina} is that we require at least order 16
in order to surpass the 1\,TFLOPS mark in double precision
(or equivalently 2\,TFLOPS in single precision), which is in SeisSol already
passed for order 5.
This mark is surpassed in 2D for all orders except 8, but in
3D only on SKX for orders 16 and 24 in double precision and
orders 16 and 32 in single precision.
The reason for the observed behaviour becomes clear when considering
cache usage in the ADER kernel (\Cref{eq:lina:derivative,eq:lina:integration}),
which contains the major part of the computational load per element:
Inside ADER we require storage for four tensors of size $\mathcal{O}^d(1+d)$,
where $\mathcal{O}=N+1$ is the order of the scheme and $d$ is the number of
spatial dimensions.
We need one buffer for the current and the next derivative,
one buffer contains the time integral, and another temporary buffer is used
for intermediate products.
In 2D, a maximum storage size of 96\,KiB is required
for all four tensors, which easily fits into L2 caches.
In 3D, we require 512\,KiB, 1.7\,MiB, and 4\,MiB for orders
16,24, and 32, and half the amount for single precision respectively.
So on KNL we need the full L2 cache per core for the four tensors when using
double precision, hence necessary cache lines are evicted from L2 cache
during ADER and we get a low arithmetic intensity.
As we only need 256\,KiB for order 16 in single precision,
the latter discussion also explains why O16 SP is more than 2$\times$
faster than O16 DP.
On SKX we have a total of 2.1375\,MiB cache per core, as the L3 cache
is non-inclusive.
Thus, only for O32 DP tensors are evicted from cache.

Moreover, we obtain an additional performance boost due to the choice of
code generator and memory layout.
For the evaluation of the flux function (i.e.\ multiplication with
$A^*, B^*,$ and $C^*$) we have to contract the degrees of freedom
with matrices smaller than $(d+1)\times (d+1)$.
We performed micro-benchmarks on KNL which suggest that the LIBXSMM
kernels are sub-optimal for such matrix sizes, which we believe to be
due to superfluous \texttt{vaddpd} instructions and due to the fixed
$8\times N$ register blocking on KNL.
As a consequence, we added PSpaMM to replace the sub-optimal kernels,
which is implemented in \yateto~by specifying a list of code generators,
where we may give higher preference to a code generator for specific
matrix sizes.
Using PSpaMM and LIBXSMM combined gives us a non-zero performance boost up
to 202\,GFLOPS in DP or up to 344\,GFLOPS in SP (for KNL O16 in 2D).
Using a CSC memory layout for the matrices
$A^*, B^*,$ and $C^*$ gives a further performance boost up
to 74\,GFLOPS in DP or up to 194\,GFLOPS in SP (for SKX O32 in 2D).
In total we obtain a performance increase of up to 31\,\% for DP
and 30\,\% for SP, and a maximum non-zero peak performance
of 1769\,GFLOPS for DP and 3969\,GFLOPS for SP.

\subsection{Findings}
In this section we found that kernels generated with \yateto~are able
to match the performance of \cite{Uphoff2017} and all major optimisations
of SeisSol are found automatically.
The extension to multiple simulations closely matches the performance of
\cite{Breuer2017}.

For LinA we saw that \yateto~does not save the developer from basic
arithmetic intensity considerations, but high performance may be achieved
with automatically generated kernels in the case of a large enough arithmetic
intensity.
Moreover, \yateto's~ability to mix and match different code generators increased
performance by up to 31\,\%.

\section{Another area of application}\label{sec:stock}
Small tensor contractions are found to be the core of
MADNESS \cite{Stock2011}, which is a framework for solving problems
in quantum chemistry and other applications \cite{Harrison2016}.
Kernels in MADNESS often take the following form \cite{Stock2011}:
\begin{equation}
  R_{ijk} = S_{xyz} X^L_{xl} X^R_{li} Y^L_{ym} Y^R_{mj} Z^L_{zn} Z^R_{nk}
  \label{eq:stock_mra}
\end{equation}
Each L-R matrix pair is obtained by a low-rank decomposition of a $p\times p$ matrix.
In the following we assume that the low-rank is equal among all L-R pairs
and given by $q$, such that indices $l,m,n$ have size $q$ and indices
$i,j,k,x,y,z$ have size $p$.
The authors of \cite{Stock2011} developed a code generator for such
tensor contractions, which uses loop transformations and outputs
loop-code and SSE3 intrinsics.
In the following we investigate the performance of the above kernel
when implemented with \yateto.

The above expression is a natural fit for \yateto's DSL.
We note that the optimal order of evaluation, as stated in
\cite{Stock2011}, is found automatically.
A few tweaks may be applied to the above formulation:
First of all, the matrices $X^L$ and $X^R$
are stored already transposed, as then a transpose-free kernel
is obtained.
(The authors of \cite{Stock2011} also store $Z^L$ transposed.)
Second, the rows of $X^L$ and $X^R$ may be padded with zeros,
such that the number of rows are a multiple of the vector width
(i.e.\ 8 in double precision and 16 in single precision.)

\begin{figure}
 \includegraphics[width=\textwidth]{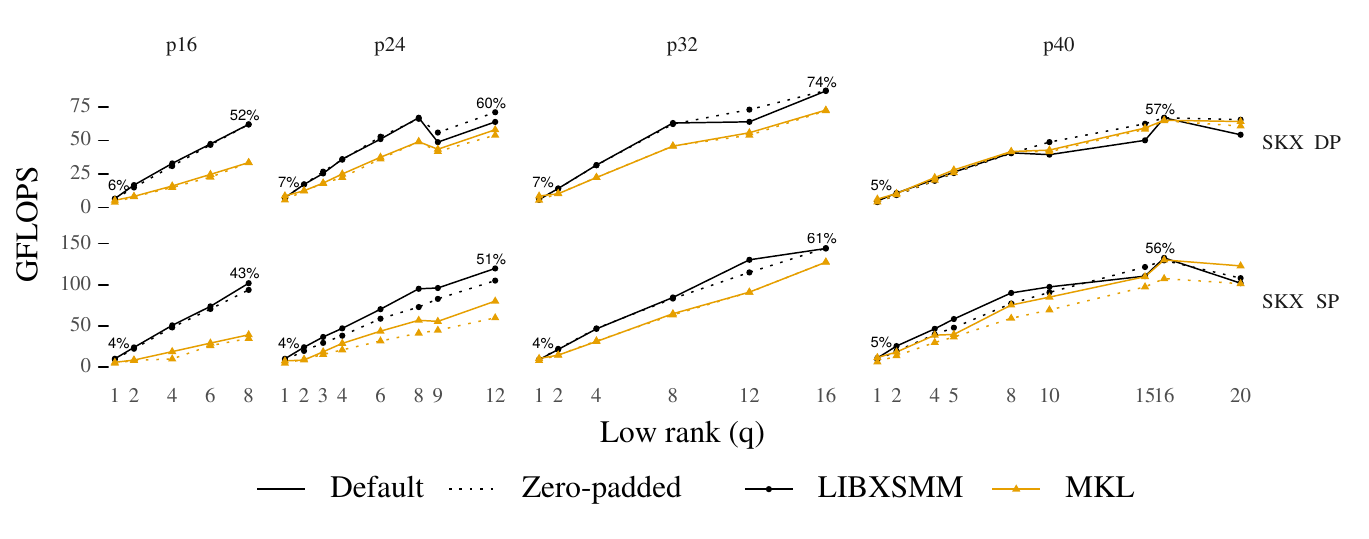}
 \caption{Performance of the kernel from \Cref{eq:stock_mra}
 for multiple tensor sizes $p$ and $q$ (see text).
 The plot shows non-zero flops, i.e.\ artificial
 flops due to zero-padding are not counted.
 The single-core AVX-512 peak performance on SKX
 is calculated w.r.t.\ the empirically determined peak frequency
 of 3.7\,GHz.}\label{fig:stock_mra}
\end{figure}

\Cref{fig:stock_mra} shows our performance results.
We observe that in the extreme case of $q=1$ we obtain
low performance of 6.4\,GFLOPS--8.8\,GFLOPS in DP
and 10.2\,GFLOPS--11.6\,GFLOPS in SP.
A similar performance was achieved in \cite{Stock2011}
for $q=1$, even though they used a CPU from the year 2008.
A potential issue is that the generated GEMMs degenerate
to either inner products, outer products, or scalar multiplications.
Moving to larger values of $q$, we see a steady increase in performance,
up to 87.3\,GFLOPS in DP and 144.6\,GFLOPS in SP, which corresponds
to 74\,\% and 61\,\% of the theoretical peak performance w.r.t.\ the
empirically determined AVX-512 peak frequency of 3.7\,GHz, respectively.

Comparing LIBXSMM to MKL we see that using LIBXSMM is beneficial in most
cases, especially for small tensor sizes.
In many cases zero-padding is either beneficial
or delivers almost the same performance as its counterpart.
But for cases p24 SP and p40 SP zero-padding mostly lowers performance.
The reason for this behaviour stems likely from the multiplication with $X^R$:
Here, the introduction of zeros leads to a non-fused GEMM,
which is fused in the non-padded variant.

In summary, we find that mapping operations to GEMM might be
sub-optimal for very small tensor contractions.
However, given large enough tensor contractions we obtain high performance
with low implementation effort, as the DSL naturally suits the problem description
and the evaluation order with minimum cost is found automatically.
Furthermore, the choice of memory layout (padded, non-padded) and GEMM back-end
(LIBXSMM, MKL) might have a large impact on performance.

The example presented here is included in \yateto's main repository
in the $\texttt{examples}$ folder.

\section{Conclusion and outlook}
We presented \yateto, a tensor toolbox especially suited
for small tensor operations.
The toolbox includes standard techniques, such as strength reduction
or the mapping of tensor operations to Loop-over-GEMM.
To our best knowledge, we present novel optimisations
for tensor operations such as equivalent sparsity patterns or
optimal index permutations.

We included \yateto~in SeisSol, which allows to generate kernels
for single forward simulations but also multiple forward simulations.
From our point of view, the DSL simplifies the simultaneous implementation
of single and multiple simulations considerably,
but still allows to tune both implementation variants individually.
Moreover, we integrated \yateto~in the tensor product based code LinA
with low implementation effort.

The performance of our new SeisSol version matches its previous performance
for a single simulation.
For multiple simulations, we obtain speed-ups of $1.1\times$ to $6.5\times$
and a non-zero peak efficiency of up to 27.5\,\% on Skylake.
In LinA, we saw that ability to mix-and-match GEMM kernels
allows to increase performance by up to 31\,\%, leading
to a non-zero peak efficiency of up to 48\,\%.
Application of \yateto~to a literature example showed high performance
over 50\,\% peak for large enough tensors, but likely sub-optimal performance
for very small tensors.

A current limitation of our approach is that the employed code-generators
for small matrix multiplications do not support transpositions.
In our test cases we are able to avoid transpositions, but this
might not be the case for other applications.
However, one may always use a BLAS implementation like MKL, which support
transpositions, and it is a simple matter to include further code generators
or libraries that offer GEMM.

In future work we want to evaluate \yateto~for more applications
and further back-ends.
Especially, we are interested in vectorisation over elements
(also sometimes called patches) for hexahedral meshes,
or in curvilinear elements.
Furthermore, we think that a useful extension to \yateto~are tools that
assist the user in the development process.
For example, a tool could automatically generate roof-line models
or could estimate the cache usage of a kernel.

\begin{acks}

The work in this paper was supported by the Volkswagen Foundation
(ASCETE -- Advanced Simulation of Coupled Earthquake-Tsunami Events, grant no.\ 88479).
Computing resources were provided by the Intel Parallel Computing Center ExScaMIC-KNL
and by the Leibniz Supercomputing Centre (project pr45fi).

\end{acks}

\bibliographystyle{ACM-Reference-Format}
\bibliography{yateto}

\newpage
\appendix
\section{SeisSol: Full performance results}\label{apx:seissol}
\begin{table}[ht]
\caption{Full results of the SeisSol benchmark,
 where NZ:=non-zero GFLOPS and HW:=hardware GFLOPS.
 $\text{HW}_\text{mad}$ shows the standard deviation
 estimated with the consistent median absolute deviation.}\label{tab:apx:seissol}

\begin{subtable}[t]{.49\linewidth}
\centering

\scalebox{0.75}{
\begin{tabular}{lllrrrr}
  \toprule
Ord. & Sim. & SMT & MDoF/s & $\text{NZ}_\text{max}$ & $\text{HW}_\text{max}$ & $\text{HW}_\text{mad}$ \\ 
  \midrule
O2 & R & Yes & 770 & 134 & 1229 & 8.4 \\ 
  O2 & 1 & Yes & 739 & 129 & 1172 & 36.8 \\ 
  O2 & 16 & Yes & 3037 & 525 & 611 & 8.5 \\ 
  O2 & 32 & Yes & 4034 & 697 & 808 & 12.1 \\ 
  O3 & R & Yes & 1502 & 324 & 1709 & 5.4 \\ 
  O3 & 1 & Yes & 1427 & 307 & 1615 & 22.6 \\ 
  O3 & 16 & Yes & 4882 & 1038 & 1224 & 12.3 \\ 
  O3 & 32 & No & 3857 & 820 & 962 & 10.4 \\ 
  O4 & R & Yes & 1902 & 536 & 2122 & 56.1 \\ 
  O4 & 1 & Yes & 1897 & 534 & 2108 & 83.2 \\ 
  O4 & 16 & No & 4061 & 1143 & 1403 & 33.5 \\ 
  O4 & 32 & No & 2961 & 833 & 1017 & 22.2 \\ 
  O5 & R & Yes & 1930 & 723 & 2549 & 55.8 \\ 
  O5 & 1 & Yes & 1949 & 730 & 2568 & 55.6 \\ 
  O5 & 16 & Yes & 3030 & 1145 & 1469 & 33.8 \\ 
  O5 & 32 & No & 2163 & 817 & 1043 & 89.7 \\ 
  O6 & R & Yes & 1494 & 752 & 2902 & 5.6 \\ 
  O6 & 1 & Yes & 1489 & 749 & 2889 & 1.4 \\ 
  O6 & 16 & No & 2221 & 1135 & 1534 & 66.7 \\ 
  O6 & 32 & Yes & 1603 & 818 & 1095 & 13.8 \\ 
  O7 & R & Yes & 1204 & 806 & 2888 & 10.4 \\ 
  O7 & 1 & Yes & 1184 & 792 & 2839 & 11.2 \\ 
  O7 & 16 & Yes & 1191 & 814 & 1142 & 10.0 \\ 
  O7 & 32 & Yes & 1285 & 877 & 1220 & 17.6 \\ 
   \bottomrule
\end{tabular}
}

\caption{KNL / SP}
\end{subtable}
\begin{subtable}[t]{.49\linewidth}
\centering

\scalebox{0.75}{
\begin{tabular}{lllrrrr}
  \toprule
Ord. & Sim. & SMT & MDoF/s & $\text{NZ}_\text{max}$ & $\text{HW}_\text{max}$ & $\text{HW}_\text{mad}$ \\ 
  \midrule
O2 & R & Yes & 687 & 120 & 548 & 12.6 \\ 
  O2 & 1 & Yes & 670 & 117 & 531 & 11.0 \\ 
  O2 & 8 & Yes & 2262 & 392 & 443 & 15.3 \\ 
  O2 & 16 & Yes & 1981 & 342 & 397 & 6.8 \\ 
  O3 & R & Yes & 1117 & 241 & 737 & 7.5 \\ 
  O3 & 1 & Yes & 1115 & 240 & 732 & 13.7 \\ 
  O3 & 8 & No & 2497 & 532 & 625 & 7.5 \\ 
  O3 & 16 & No & 2295 & 488 & 573 & 29.5 \\ 
  O4 & R & Yes & 1211 & 341 & 1154 & 14.1 \\ 
  O4 & 1 & Yes & 1184 & 334 & 1126 & 5.7 \\ 
  O4 & 8 & No & 2122 & 598 & 742 & 1.9 \\ 
  O4 & 16 & Yes & 1842 & 518 & 637 & 18.6 \\ 
  O5 & R & Yes & 1169 & 438 & 1314 & 21.4 \\ 
  O5 & 1 & Yes & 1171 & 438 & 1314 & 14.5 \\ 
  O5 & 8 & Yes & 1579 & 598 & 775 & 16.4 \\ 
  O5 & 16 & No & 1370 & 518 & 664 & 16.8 \\ 
  O6 & R & Yes & 917 & 461 & 1435 & 6.9 \\ 
  O6 & 1 & Yes & 912 & 459 & 1426 & 9.0 \\ 
  O6 & 8 & No & 1134 & 581 & 792 & 13.2 \\ 
  O6 & 16 & Yes & 876 & 447 & 601 & 20.9 \\ 
  O7 & R & No & 614 & 411 & 1343 & 79.9 \\ 
  O7 & 1 & No & 605 & 405 & 1323 & 112.0 \\ 
  O7 & 8 & Yes & 623 & 417 & 546 & 10.4 \\ 
  O7 & 16 & Yes & 673 & 459 & 645 & 11.5 \\ 
   \bottomrule
\end{tabular}
}

\caption{KNL / DP}
\end{subtable}

\begin{subtable}[t]{.49\linewidth}
\centering

\scalebox{0.75}{
\begin{tabular}{lllrrrr}
  \toprule
Ord. & Sim. & SMT & MDoF/s & $\text{NZ}_\text{max}$ & $\text{HW}_\text{max}$ & $\text{HW}_\text{mad}$ \\ 
  \midrule
O2 & R & Yes & 510 & 89 & 815 & 3.8 \\ 
  O2 & 1 & Yes & 530 & 92 & 840 & 0.3 \\ 
  O2 & 16 & Yes & 2939 & 508 & 584 & 0.9 \\ 
  O2 & 32 & Yes & 3300 & 570 & 653 & 3.9 \\ 
  O3 & R & Yes & 1249 & 269 & 1421 & 3.5 \\ 
  O3 & 1 & Yes & 1344 & 290 & 1521 & 10.5 \\ 
  O3 & 16 & Yes & 3377 & 718 & 877 & 3.2 \\ 
  O3 & 32 & Yes & 3619 & 769 & 935 & 1.6 \\ 
  O4 & R & Yes & 1562 & 440 & 1741 & 6.0 \\ 
  O4 & 1 & Yes & 1542 & 434 & 1714 & 26.8 \\ 
  O4 & 16 & Yes & 3565 & 1003 & 1306 & 2.2 \\ 
  O4 & 32 & No & 3687 & 1037 & 1343 & 0.5 \\ 
  O5 & R & Yes & 1901 & 712 & 2511 & 21.0 \\ 
  O5 & 1 & Yes & 1877 & 703 & 2474 & 12.4 \\ 
  O5 & 16 & Yes & 3655 & 1382 & 1918 & 1.5 \\ 
  O5 & 32 & Yes & 3712 & 1402 & 1937 & 20.2 \\ 
  O6 & R & No & 1965 & 988 & 3816 & 21.7 \\ 
  O6 & 1 & No & 1982 & 997 & 3845 & 49.8 \\ 
  O6 & 16 & Yes & 3600 & 1840 & 2729 & 44.5 \\ 
  O6 & 32 & Yes & 3544 & 1809 & 2660 & 58.4 \\ 
  O7 & R & No & 1827 & 1222 & 4383 & 36.0 \\ 
  O7 & 1 & No & 1802 & 1206 & 4321 & 42.8 \\ 
  O7 & 16 & Yes & 3275 & 2237 & 3493 & 68.0 \\ 
  O7 & 32 & No & 3206 & 2187 & 3389 & 8.5 \\ 
   \bottomrule
\end{tabular}
}

\caption{SKX / SP}
\end{subtable}
\begin{subtable}[t]{.49\linewidth}
\centering

\scalebox{0.75}{
\begin{tabular}{lllrrrr}
  \toprule
Ord. & Sim. & SMT & MDoF/s & $\text{NZ}_\text{max}$ & $\text{HW}_\text{max}$ & $\text{HW}_\text{mad}$ \\ 
  \midrule
O2 & R & Yes & 414 & 72 & 330 & 0.7 \\ 
  O2 & 1 & Yes & 427 & 74 & 339 & 0.2 \\ 
  O2 & 8 & Yes & 1402 & 243 & 274 & 0.2 \\ 
  O2 & 16 & Yes & 1609 & 278 & 318 & 0.9 \\ 
  O3 & R & Yes & 696 & 150 & 459 & 2.1 \\ 
  O3 & 1 & Yes & 696 & 150 & 457 & 1.2 \\ 
  O3 & 8 & Yes & 1659 & 353 & 415 & 0.6 \\ 
  O3 & 16 & Yes & 1788 & 380 & 462 & 0.6 \\ 
  O4 & R & Yes & 1043 & 294 & 994 & 10.3 \\ 
  O4 & 1 & Yes & 1043 & 294 & 992 & 3.9 \\ 
  O4 & 8 & Yes & 1770 & 499 & 619 & 0.8 \\ 
  O4 & 16 & No & 1827 & 514 & 669 & 0.2 \\ 
  O5 & R & Yes & 1212 & 454 & 1363 & 5.9 \\ 
  O5 & 1 & Yes & 1199 & 449 & 1347 & 17.3 \\ 
  O5 & 8 & No & 1807 & 684 & 887 & 2.9 \\ 
  O5 & 16 & No & 1835 & 694 & 963 & 0.8 \\ 
  O6 & R & Yes & 1259 & 633 & 1970 & 28.6 \\ 
  O6 & 1 & Yes & 1244 & 626 & 1946 & 42.5 \\ 
  O6 & 8 & No & 1784 & 914 & 1246 & 3.1 \\ 
  O6 & 16 & No & 1766 & 902 & 1332 & 4.0 \\ 
  O7 & R & Yes & 1048 & 702 & 2293 & 101.3 \\ 
  O7 & 1 & Yes & 1024 & 686 & 2241 & 35.7 \\ 
  O7 & 8 & No & 1707 & 1142 & 1495 & 27.0 \\ 
  O7 & 16 & No & 1614 & 1103 & 1722 & 23.9 \\ 
   \bottomrule
\end{tabular}
}

\caption{SKX / DP}
\end{subtable}

\end{table}

\end{document}